\documentclass[12pt]{article}%
\usepackage{amsmath,amssymb,amsthm,amsfonts}
\usepackage{wasysym}
\usepackage{graphicx}
\usepackage{xcolor}
\usepackage{stackengine}

\newcounter{subfigure}
\usepackage[colorlinks]{hyperref}

\usepackage{geometry}
\usepackage{appendix}

\newcommand{\R}{\mathbb{R}}

\newcommand{\N}{\mathbb{N}}
\renewcommand{\S}{\mathbb{S}^1}

\newcommand{\ea}{\textit{et al. }}
\renewcommand{\epsilon}{\varepsilon}
\renewcommand{\th}{\text{th}}
\newcommand{\sgn}{\operatorname{sgn}}

\newtheorem{thm}{Theorem}
\newtheorem{lemma}{Lemma}

\begin{document}

\title{Standard map-like models for single and multiple walkers in an annular cavity}

\author{Aminur Rahman \thanks{Corresponding Author, \url{amin.rahman@ttu.edu},
Department of Mathematics and Statistics, Texas Tech University}}

\date{}

\maketitle

\begin{abstract}
Recent experiments on walking droplets in an annular cavity showed the existence of
complex dynamics including chaotically changing velocity.  This article presents models,
influenced by the kicked rotator/standard map, for both single and multiple droplets.  The models are shown
to achieve both qualitative and quantitative agreement with the experiments, and makes
predictions about heretofore unobserved behavior.  Using dynamical systems techniques
and bifurcation theory, the single droplet model is analyzed to prove dynamics suggested by the numerical
simulations.
\end{abstract}

\bigskip
\bigskip
\bigskip
\bigskip

\textbf{
In its simplest form, a droplet walking in an annular cavity can be thought of as a dynamical
system on a topological circle ($\S$).  As experiments have shown, while there is some radial
motion, most of the interesting dynamics is restricted to the circumferential direction.  Further, since
dynamics is of interest, it is reasonable to focus on the walking regime ignoring the transition
from bouncing to walking.  When the amplitude of acceleration of the vibrating fluid bath is
low enough the droplet bounces in place.  Once the amplitude is above a certain threshold
the droplet starts walking, but with low enough inertia the speed is damped to a steady-state
regardless of the initial velocity.  As the damping is
reduced, the velocity destabilizes and eventually leads to chaotic walking.  Interestingly,
destabilization of the velocity also occurs in the large damping regime when another droplet
is added, then the $n+1$ droplets reach an equal steady state velocity higher than that
of $n$ droplets.  One well-known dynamical system with similar behavior is a \emph{kicked rotator}
\cite{KickedRotator} (also called a \emph{kicked rotor}), the dynamics of which can be described
using the \emph{standard map} \cite{Chirikov1969, Chirikov1979}.
Of course, the physical limitations of how far a droplet
can move in one leap disqualifies the standard map as a candidate for a model.
However, models influenced by the standard map, for both individual and multiple
droplet dynamics in an annulus are developed in this investigation.  The single droplet model
is then analyzed through numerical
simulations and dynamical systems theory.  Simulations show close agreement with experiments.
Analysis predicts behavior suggested by the simulations.  The combination of the
two make for a complete dynamical systems study of droplets in an annulus.}

\section{Introduction}
\label{Sec: Intro}

In the \emph{Annual Review of Fluid Mechanics}, Bush \cite{Bush15a} highlights the important role of
dynamical systems in studying walking droplets.  While there has been an abundance of
continuous finite and infinite dimensional dynamical or semi-dynamical systems models,
the introduction of discrete dynamical (in the strictest sense of the word) models has been a
recent occurrence.  A two-dimensional billiards - type model was developed by Shirokoff
\cite{Shirokoff13}, then Gilet \cite{Gilet14} presented perhaps the simplest dynamical model,
which considers straight-line walking in a semi-confined geometry.  Later, Gilet \cite{Gilet16}
derived a more complex, in both the philosophical and mathematical sense, discrete model
for walkers on a circular corral.

In recent years, studying simpler versions of the well-known walking droplets phenomenon has
garnered much interest.  There have been experiments \cite{FHSV17}, dynamical models \cite{Gilet14},
and dynamical systems analysis \cite{RahmanBlackmore16} on walking in a straight-line (linear channels).
The dynamical models and early analysis also inspired the discovery of a new type of homoclinic-like
global bifurcation \cite{RahmanBlackmore17}, which lead to a study of the dynamics of a generalized
map exhibiting this new bifurcation \cite{RJB17}.

Walking on a straight-line in a confined geometry will admit dynamics on a finite interval $[a,b] \in \R$.
However, boundaries pose nontrivial modeling issues as shown by Gilet \cite{Gilet14} and by Rahman
and Blackmore \cite{RahmanBlackmore16}, where the droplet escapes the confinement when the
damping is dropped (or the inertia is increased) beyond the threshold for chaos.  One way to remedy
this issue is to extend the domain to the whole real line, however that makes it impossible to compare
certain statistical properties, such as histograms, with that of experiments.

Alas, we call upon dynamical systems and topological techniques to aid us.  Suppose there exists a
$C^n$ - discrete dynamical system on an interval $[a,b]$ defined as $x_{n+1} = f(x_n)$ and there
exists an $M$ such that $f(b - \omega) > b$ and $f(a + \omega) < a$ for all $0 < \omega < M$.  Clearly,
the iterates are escaping the domain $[a,b]$, however if the two endpoints can be identified ( i.e.
$f(a) = f(b),\, f'(a) = f'(b),\,\ldots,\, f^{(n)}(a)=f^{(n)}(b)$ and $x_{n+1} = a + f(x_n) \mod (b-a)$),
then we have a dynamical system confined to $[a,b]$ without any boundary issues.  Moreover, the system
can be thought of as being on a circle $\S$ by transforming it into the system $\theta_{n+1} = g(\theta_n)$, which is
the simplest mathematical analog of a walker on an annulus.

It is fortunate that there have been some experiments with walkers on an annulus.  It has been
demonstrated by Filoux \ea \cite{FHV15} that a single droplet walking in the ``low memory
regime'' ($95\%$ of Faraday wave threshold) has a constant speed, which is independent of the
circumference.  When another droplet is inserted, the speed increases to a different steady-state.
This increase in speed occurs as more and more droplets enter the annulus.  On the other end of
the memory spectrum, Pucci and Harris conducted experiments for a single droplet above the Faraday
wave threshold in which they observed the walker's speed destabilizing and becoming seemingly chaotic
\cite{PucciHarrisPrivate}.

One may suggest transforming Gilet's model \cite{Gilet14} to be on an annulus, however
it does not reproduce behavior observed in experiments because the model implies that the walker
is stagnant at certain fixed points in the chaotic regime.  This necessitates a novel model
for dynamics on an annulus, which is the focus of this article, beginning in Sec. \ref{Sec: Single} with
models of a single droplet.  The dynamical systems and bifurcation analysis of the model from Sec. \ref{Sec: Single}
is presented in Sec. \ref{Sec: Analysis}.  Then in Sec. \ref{Sec: Multiple} the standard map-like model from
Sec. \ref{Sec: Single} is extended to also include multiple droplets.  In both Sec. \ref{Sec: Multiple} and
\ref{Sec: Single} we show quite
close quantitative agreement with experiments of multiple droplets \cite{FHV15} and qualitative agreement for
chaotic walking of a single droplet \cite{PucciHarrisPrivate}.  Section \ref{Sec: Combined} combines the two
models for a speculative unified model of any number of droplets in an annulus and conjectures the behavior for
multiple droplets in the chaotic regime.  Finally, in Sec. \ref{Sec: Conclusion} we close on some remarks about how
such a simple model can describe such complex phenomena.

\section{Modeling of a single droplet}\label{Sec: Single}

In the standard map\cite{Chirikov1969, Chirikov1979},
\begin{equation}
\begin{split}
p_{n+1} &= p_n + K\sin(\theta_n) \mod 2\pi,\\
\theta_{n+1} &= \theta_n + p_{n+1} \mod 2\pi;
\end{split}
\label{Eq: StandardMap}
\end{equation}
a ``particle'' receives a ``kick'' at each discrete timestep $n$, where $\theta$ represents the
position of the particle and $p$ represents the momentum, and $K$ is the kicking strength.
The map becomes chaotic when $K$ is near unity.

Similarly, a droplet on an annulus receives a kick at each impact with its wavefield.  Unlike
the standard map, the amplitude of the kick will be inversely proportional to the distance
traveled since the wavefield exhibits spatial and temporal decay.  Further, since the experiments
of Filoux \ea \cite{FHV15} measure changes in velocity, we are interested in modeling the
velocity rather than the momentum.  While the model of Filoux \ea \cite{FHV15} is not a
dynamical system, it gives us a good starting point for the velocity contributions from
impacts with the droplet's wavefield.
They write the velocity equation for multiple droplets as
\begin{equation}
\begin{split}
&u_{n+1}^i - u_n^i = -\gamma u_n^i - C_0\frac{\partial \zeta^{ii}}{\partial s}\bigg|_{n+1}
- C_1\sum_{j\neq i} \frac{\partial \zeta{ij}}{\partial s}\bigg|_{n+1},\\
&\zeta^{ij} = \zeta_0\cos\left(\frac{2\pi(s_{n+1}^i - s_n^j)}{\lambda_F}\right)
\exp\left(-\frac{s_{n+1}^i - s_n^j}{C_2\lambda_F}\right);
\end{split}
\end{equation}
where $u_n^i$ and $s_n^i$ are the velocity and position of the $i^\th$ droplet at the
$n^\th$ impact.  Damping of the velocity contribution from the previous time step is
represented by $\gamma$. The parameter $C_0$ is the contribution from the $i^\th$
droplet to its own velocity and $C_1$ is the contribution from all other droplets to the
$i^\th$ droplet's velocity. The wavefield is given by the second equation, where $\zeta$
is the wave profile, $\lambda_F$ is the Faraday wavelength, and $C_2$ is a constant
used as a tuning parameter.

The model in this article is similar in that it uses a sine term for the shape of the velocity
profile and an exponential term for its damping, however our argument in the exponent is
squared instead of an absolute value since the absolute value does not have a continuous
derivative.  Unlike the model of Filoux \ea \cite{FHV15},
where they were interested in the steady state of the velocity, this one is a dynamical system
that describes the changes in velocity.  Like the standard map \cite{Chirikov1969, Chirikov1979}
it is assumed that the droplet receives a ``kick'' at each impact $n$.  Other assumptions include, a constant
damping, $C \in [0,1]$, of the velocity due to the inertia of the droplet\cite{Shirokoff13, Gilet14, Gilet16},
and a constant undamped kick strength of $K\in \R^+$.  Furthermore, we use nondimensional variables
and parameters from the onset.  Now we may write the model for velocity,
\begin{equation}
v_{n+1} = C\left[v_n + K\sin(\omega(\theta_n - \theta_{n-1}))e^{-\nu (\theta_n - \theta_{n-1})^2}\right],
\label{Eq: SingleVelocity}
\end{equation}
where $v_n$ and $\theta_n$ are the velocity and positions of the single droplet after the $n^\th$ impact.
The nondimensional frequency is $\omega = 2\pi(R_{\text{in}} + D/2)/\lambda_F$\cite{FHV15}, where
$\lambda_F \approx 6 mm$ is the Faraday wavelength and $R_{\text{in}} + D/2
= 13.75 mm, 41.25 mm, 68.75 mm$ is the average of the inner and outer radii from the experiments of
Filoux \ea \cite{FHV15}.  The spatial damping of the kick strength is represented by the
exponential term with a damping parameter $\nu\in \R^+$, which implicitly depends on $\lambda_F$
and the memory of the system.  For the simulations, $\omega$ is fixed in the experimental range of
Filoux \ea\cite{FHV15}.  For the argument in the exponent, if we nondimensionalize that of Filoux \ea
and evaluate it at a spacing of $\lambda_F$, we get $-1/2.1$\cite{FHV15}.  In order to reproduce
this in our formulation we let $\nu = \omega^2/2.1\cdot 4\pi^2$.

As in the standard map \eqref{Eq: StandardMap}, $\theta_{n+1} = \theta_n + v_{n+1}$
so the right hand side of \eqref{Eq: SingleVelocity} can be represented by the function
$f:\S\rightarrow\S$ defined as
\begin{equation}
f(v) := C\left[v + K\sin(\omega v)e^{-\nu v^2}\right].
\label{Eq: SingleMap}
\end{equation}
The amplitude of $f$ should decay as $v$ increases; i.e.
the farther a droplet travels on the previous bounce, the smaller its displacement
will be on the next bounce.  Further, in order for the model to be a dynamical
system on $\S$, the end points $-\pi$ and $\pi$ need to be identified (i.e., $f(-\pi) = f(\pi)$
and $f'(-\pi) = f'(\pi)$), which
is done first by observing
\begin{equation*}
\frac{f(-\pi)}{C} = -\pi + K\sin(-\pi\omega)e^{-\nu\pi^2} = -\pi - K\sin(\pi\omega)e^{-\nu\pi^2}
\end{equation*}
and
\begin{equation*}
\frac{f(\pi)}{C} = \pi + K\sin(\pi\omega)e^{-\nu\pi^2},
\end{equation*}
then
\begin{equation}
K = \frac{-\pi e^{\nu\pi^2}}{\sin(\pi\omega)}.
\end{equation}
We notice that this choice of $K$ gives $f(-\pi) = f(\pi) = 0$, which ensures that the amplitude
of $f$ will be a decreasing
function that vanishes to zero since it is continuous.
Next the full model can be written as
\begin{equation}
\begin{split}
v_{n+1} &= C\left[v_n -\pi\frac{\sin(\omega v_n)}{\sin(\pi\omega)}e^{(\omega^2/8.4\pi^2)(\pi^2-v_n^2)}\right],\\
\theta_{n+1} &= \theta_n + v_{n+1}.
\end{split}
\label{Eq: SingleModel}
\end{equation}
Then the only free parameter is $C \ll 1$.  If $C = 0$ the droplet remains in one spot, however
if $C \sim 1$, the kick strength is too large and the droplet travels farther
than physically possible.  It is observed that $C \sim 1/K$ is ideal for reproducing the physical
phenomena. 
For this choice of $C$, there is a certain amount of
momentum preserved from the previous impacts, however the largest contribution is from the kick
of the current impact, which is also a reasonable physical implication.

We note that the entire system \eqref{Eq: SingleModel} only has fixed points when $v = 0$;
in which case, there exists a line of fixed points $\theta_* = \theta$ for all $\theta \in \S$.
From a dynamical systems point of view, it is more interesting to study the one-dimensional system for
velocity \eqref{Eq: SingleMap}.  While the fixed points themselves cannot be found analytically,
they can be shown to exist and computed numerically.  Furthermore, we can develop some
abstractions to be employed in the analysis of the model.

\subsection{Simulations and comparisons}

For each of the simulations we use $\omega = 31/2$, which is within the experimental range of
Filoux \ea \cite{FHV15}.  Further, we write $C$ in terms of $K$ in order to facilitate comparisons
and future asymptotic analysis.  For $C \ll 1/K$ we observe quite regular behavior, but as $C$
increases to $C \sim 1/K$, the system becomes seemingly chaotic.  This is to be expected since
the standard map \eqref{Eq: StandardMap}\cite{Chirikov1969, Chirikov1979} starts to become
chaotic when the kick is near unity.  In the model presented here \eqref{Eq: SingleModel}, the
corresponding kick strength is $C\cdot K$.

First we notice that the speed goes to zero, as expected, if the damping of the velocity is large enough,
shown in Fig. \ref{Fig: Vanishing}, where $C = 1/100K$.  

\begin{figure}[htbp]
\centering

\stackinset{r}{1mm}{t}{1mm}{\textbf{\large (a)}}{\includegraphics[width = 0.44\textwidth]{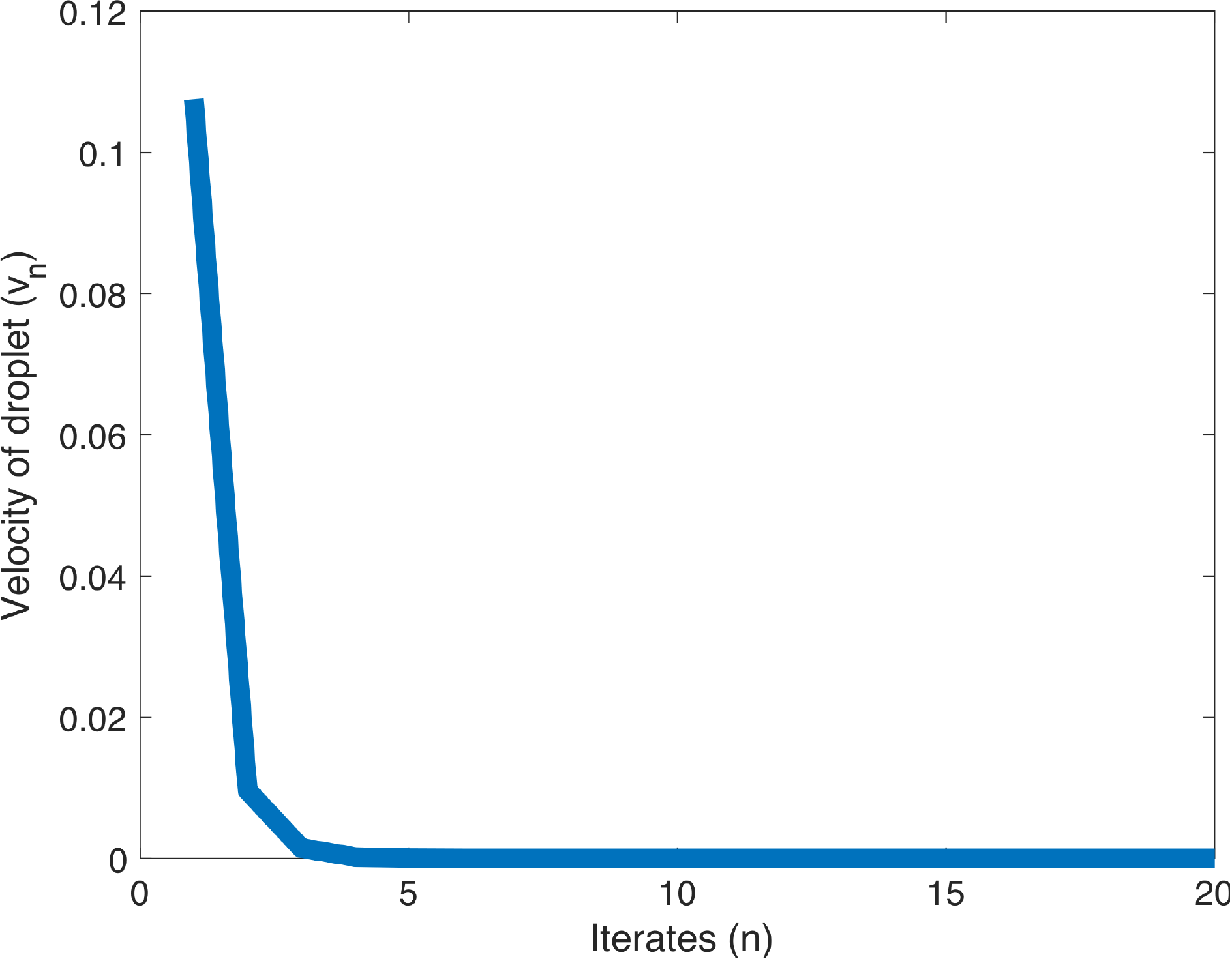}}
\refstepcounter{subfigure}\label{Fig: VanishingSpeed}
\stackinset{r}{1mm}{b}{6mm}{\textbf{\large (b)}}{\includegraphics[width = 0.44\textwidth]{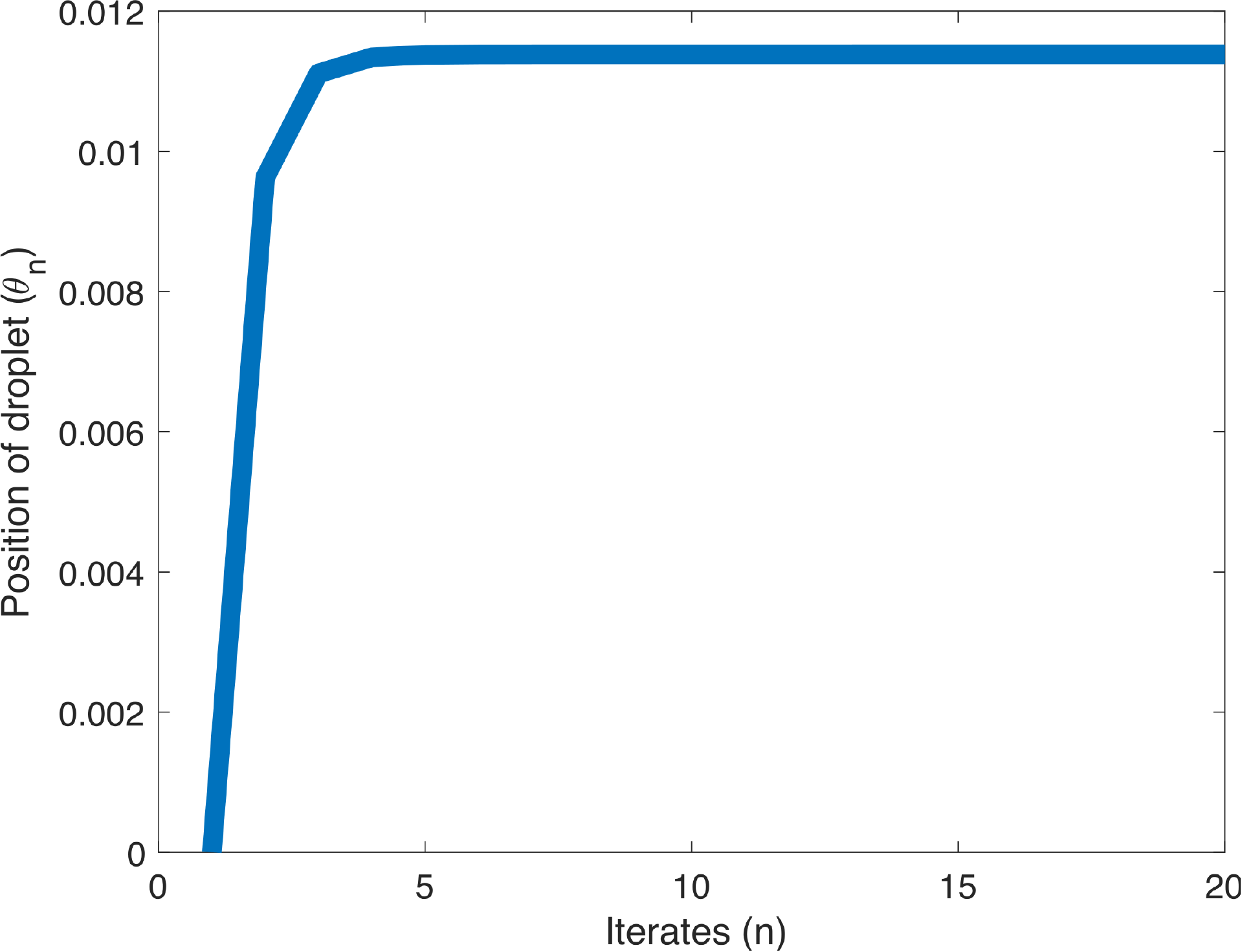}}
\refstepcounter{subfigure}\label{Fig: VanishingTheta}

\caption{\textbf{(a)}  The velocity converging to zero when $C = 1/100K$.  \textbf{(b)}  As the velocity
decreases to zero, the droplet gets stuck at a fixed location in $\theta$.}
\label{Fig: Vanishing}
\end{figure}

Then as $C$ is increased the droplet starts walking with constant speed, shown in
Fig. \ref{Fig: Constant}, where $C = 1/10K$.  This is also reported in experiments
of a single walker on an annulus below the Faraday wave threshold \cite{FHV15, MilesPrivate}.

\begin{figure}[htbp]
\centering

\stackinset{r}{1mm}{t}{1mm}{\textbf{\large (a)}}{\includegraphics[width = 0.448\textwidth]{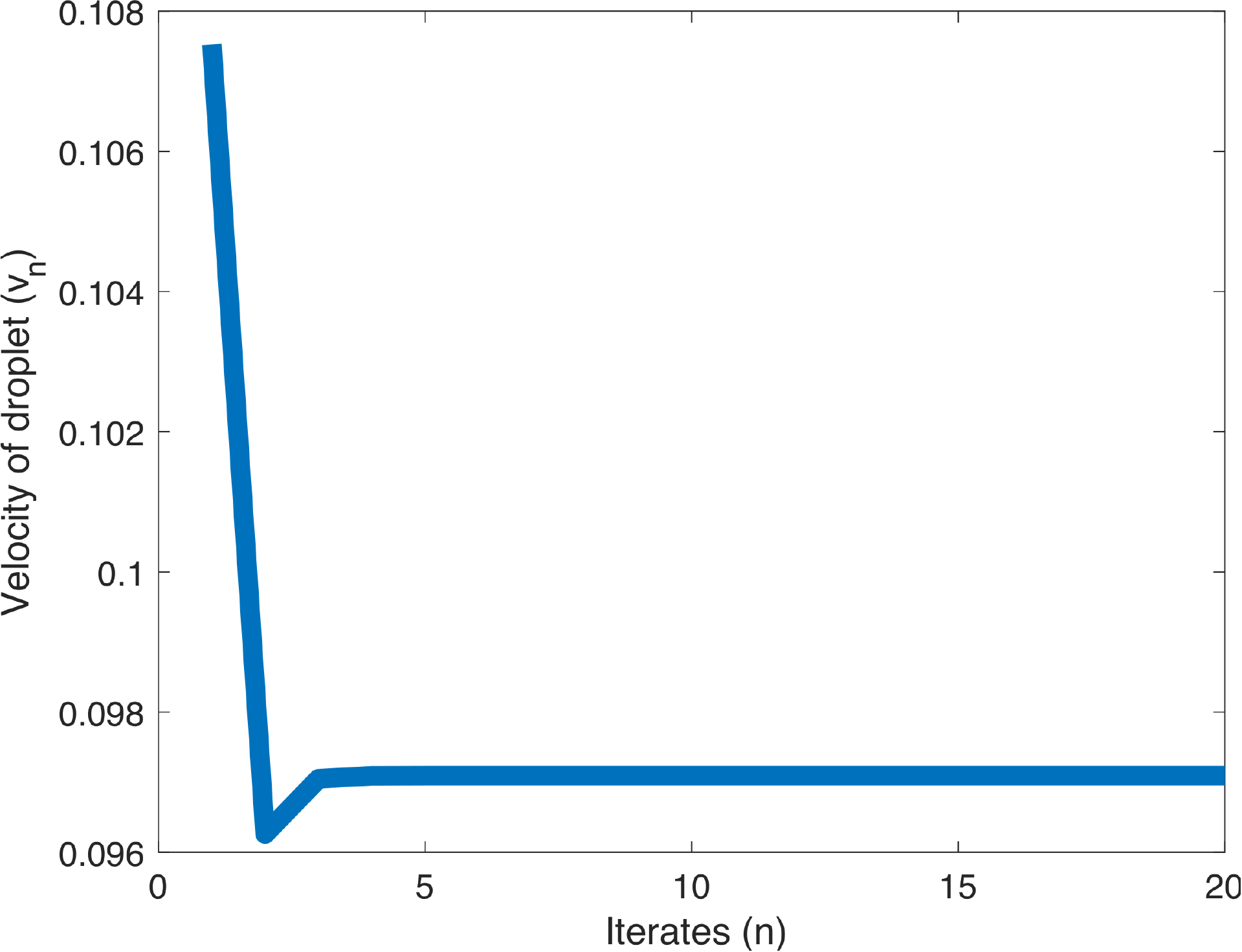}}
\refstepcounter{subfigure}\label{Fig: ConstantSpeed}
\stackinset{l}{6mm}{t}{1mm}{\textbf{\large (b)}}{\includegraphics[width = 0.432\textwidth]{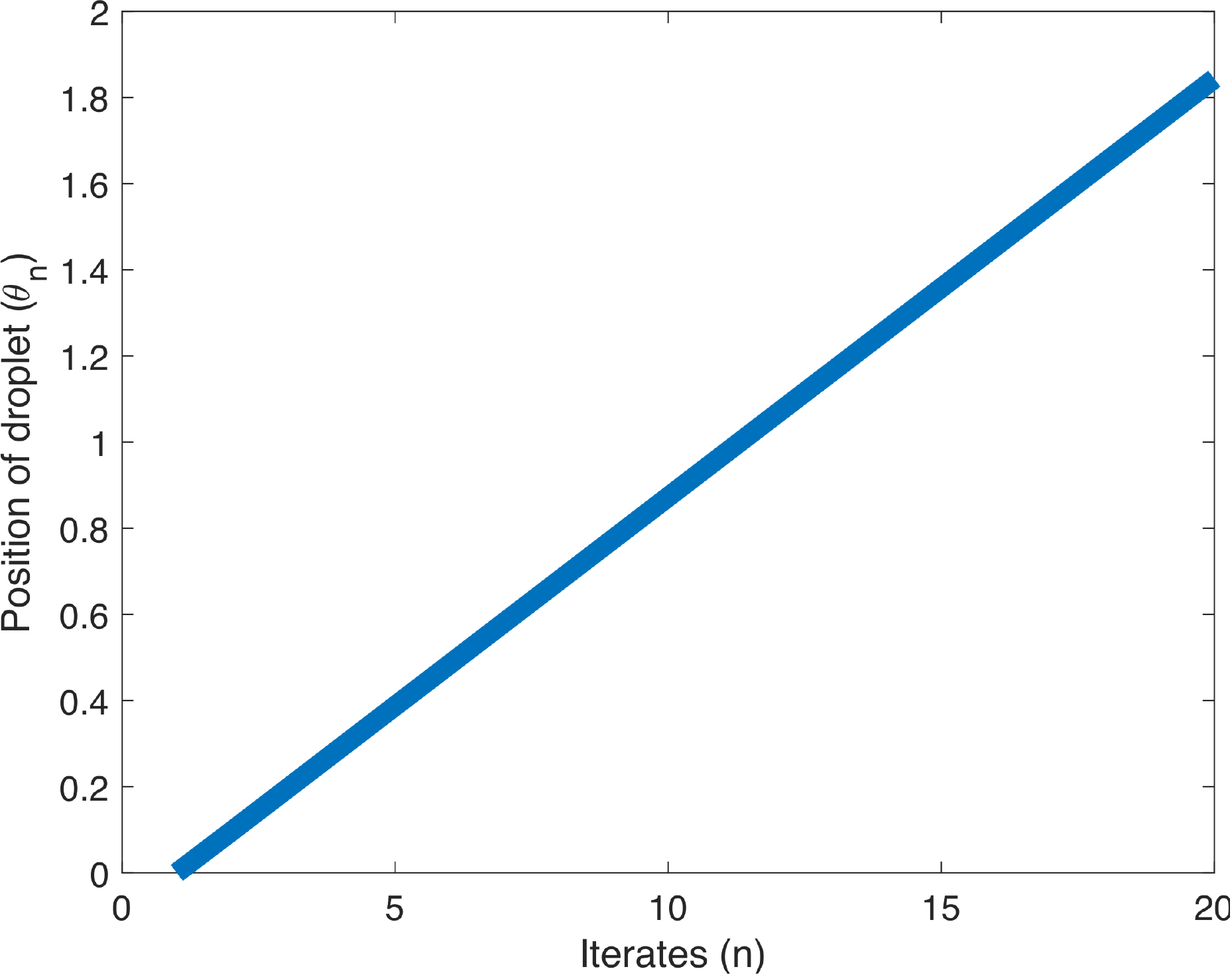}}\refstepcounter{subfigure}\label{Fig: ConstantTheta}

\caption{\textbf{(a)}  The velocity converging to a steady-state when $C = 1/10K$.
\textbf{(b)}  As the velocity converges to a constant, the walker moves around the annulus
\textit{ad infinitum}.}
\label{Fig: Constant}
\end{figure}

The truly interesting behavior arises when $C$ is increased to near $1/K$.  In Fig. \ref{Fig: ChaoticShortTime},
nine iterates are shown to illustrate the sudden changes in direction.  These types of direction changes continue
as the walker traverses the entire domain creating seemingly dense orbits on $\S$ in Fig. \ref{Fig: ChaoticLongTime}.
The behavior has also been observed in the recent experiments of Pucci and Harris \cite{PucciHarrisPrivate}.
While Figs. \ref{Fig: ChaoticShortTime} and \ref{Fig: ChaoticLongTime} show irregular behavior of the walker,
the real evidence of chaos is illustrated in Fig. \ref{Fig: ChaoticTimeSeries}, which shows the seemingly random
changes in velocity of the walker.

\begin{figure}[htbp]
\centering
\includegraphics[width = 0.9\textwidth]{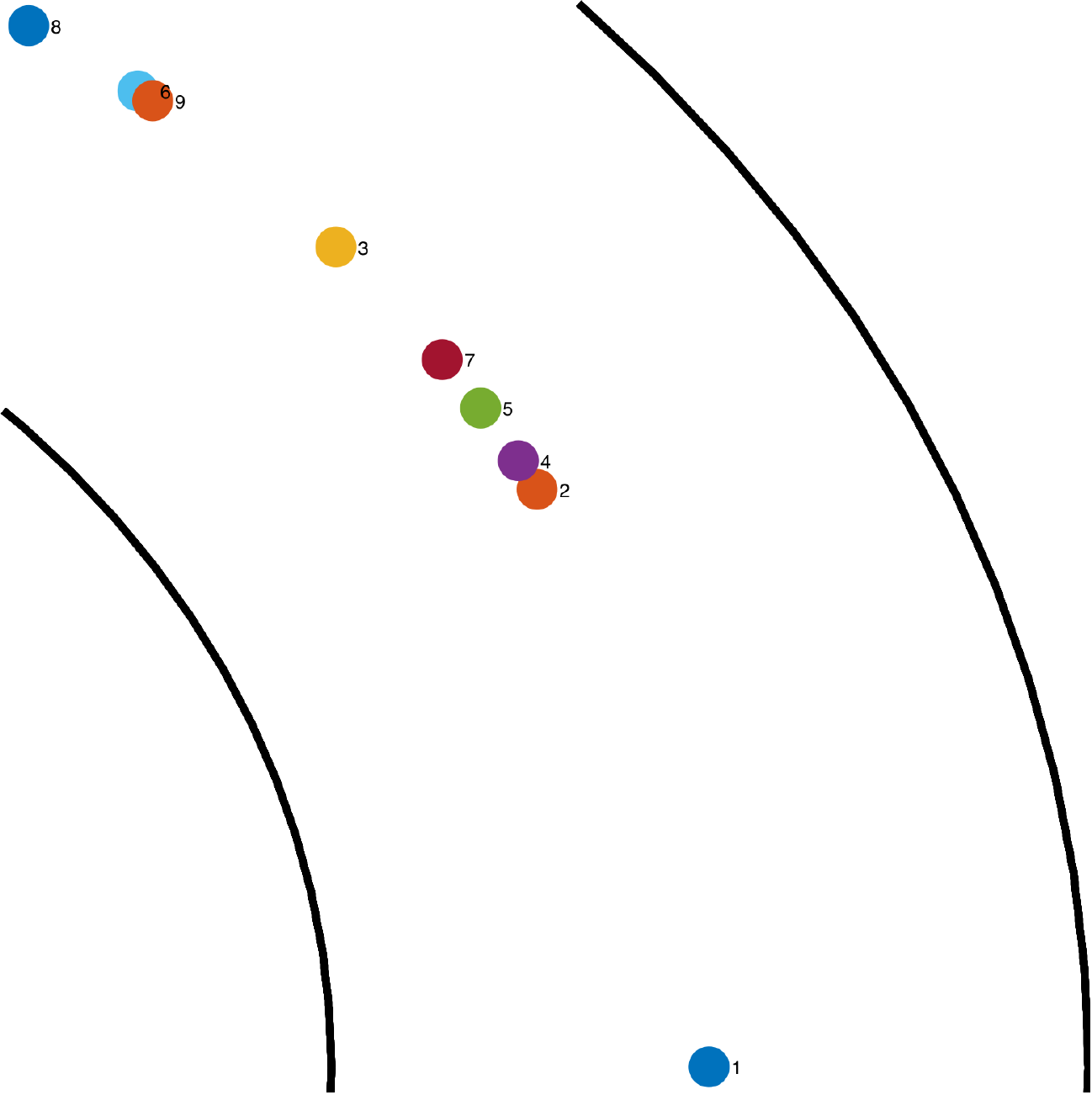}
\caption{Closeup of nine iterates of \eqref{Eq: SingleModel} representing the short time dynamics of
the droplets for $C = 1/2K$.  The iterate number is to the right of each marker.}
\label{Fig: ChaoticShortTime}
\end{figure}

\begin{figure}[htbp]
\centering

\stackinset{l}{1mm}{t}{1mm}{\textbf{\large (a)}}{\includegraphics[width = 0.4\textwidth]{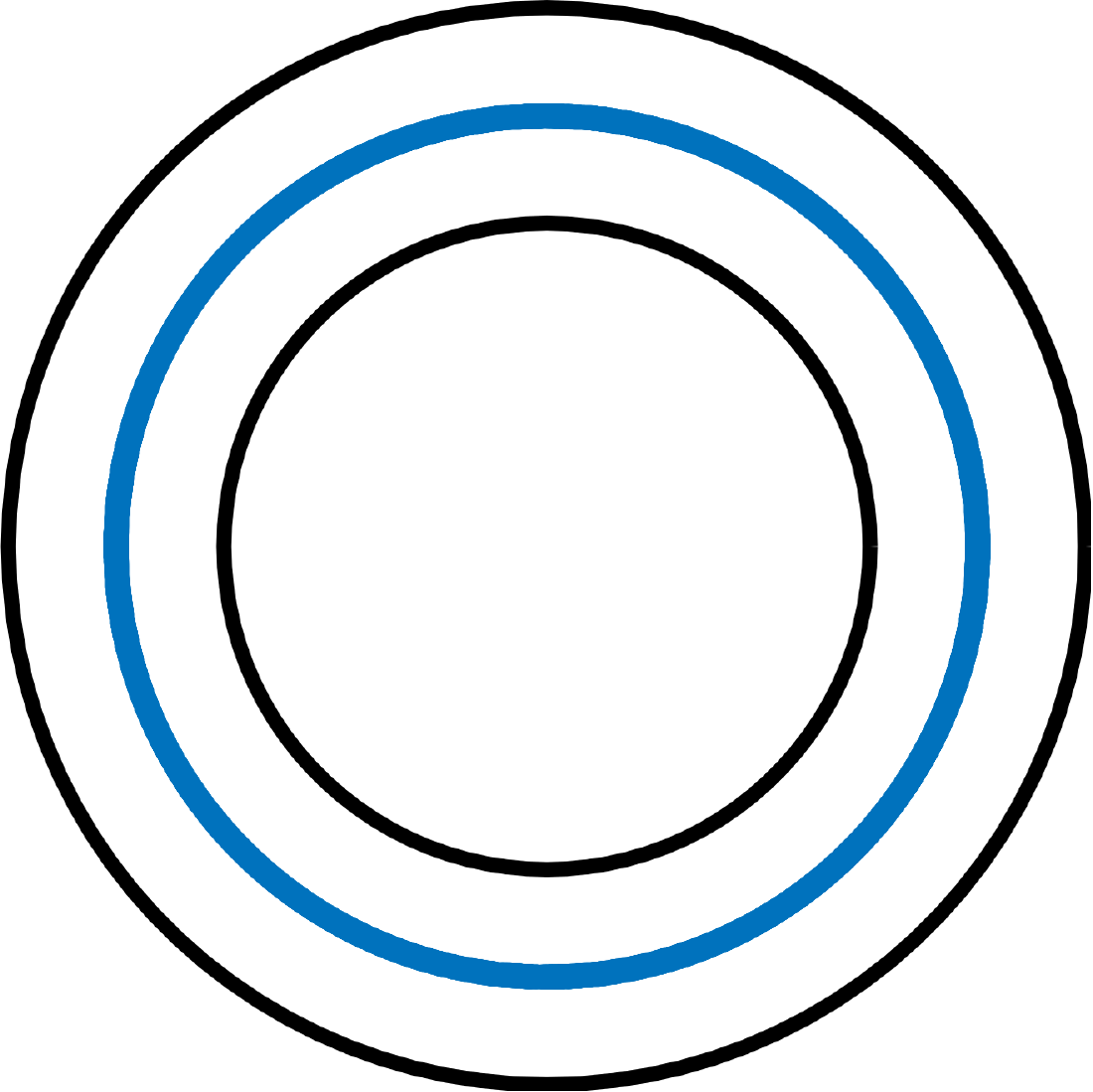}}
\refstepcounter{subfigure}\label{Fig: LongTimeChaoticSpeed}
\stackinset{r}{0.25mm}{t}{4mm}{\textbf{\large (b)}}{\includegraphics[width = 0.48\textwidth]{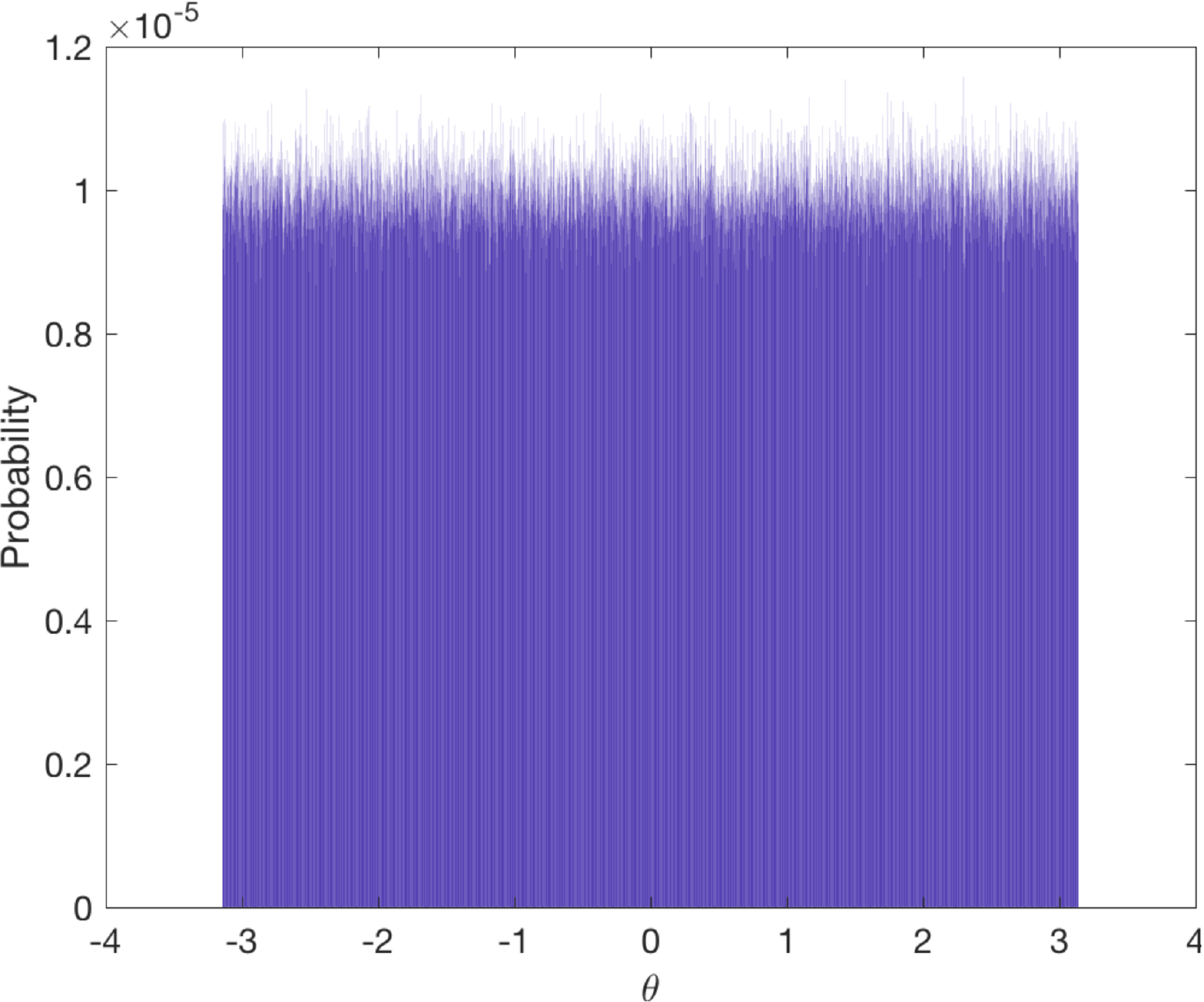}}
\refstepcounter{subfigure}\label{Fig: HistogramChaoticSpeed}

\caption{Observation of seemingly dense orbits in $\S$, which remains to be proved.
\textbf{(a)}  Long time dynamics of \eqref{Eq: SingleModel} for $C = 1/K$ after $10^8$ impacts.  
\textbf{(b)}  The probability of a walker being present in a certain interval $[\alpha,\beta]$ for the histogram is
taken at each interval of size $2\pi/10^5$.}
\label{Fig: ChaoticLongTime}
\end{figure}

\begin{figure}[htbp]
\centering
\includegraphics[width = 0.9\textwidth]{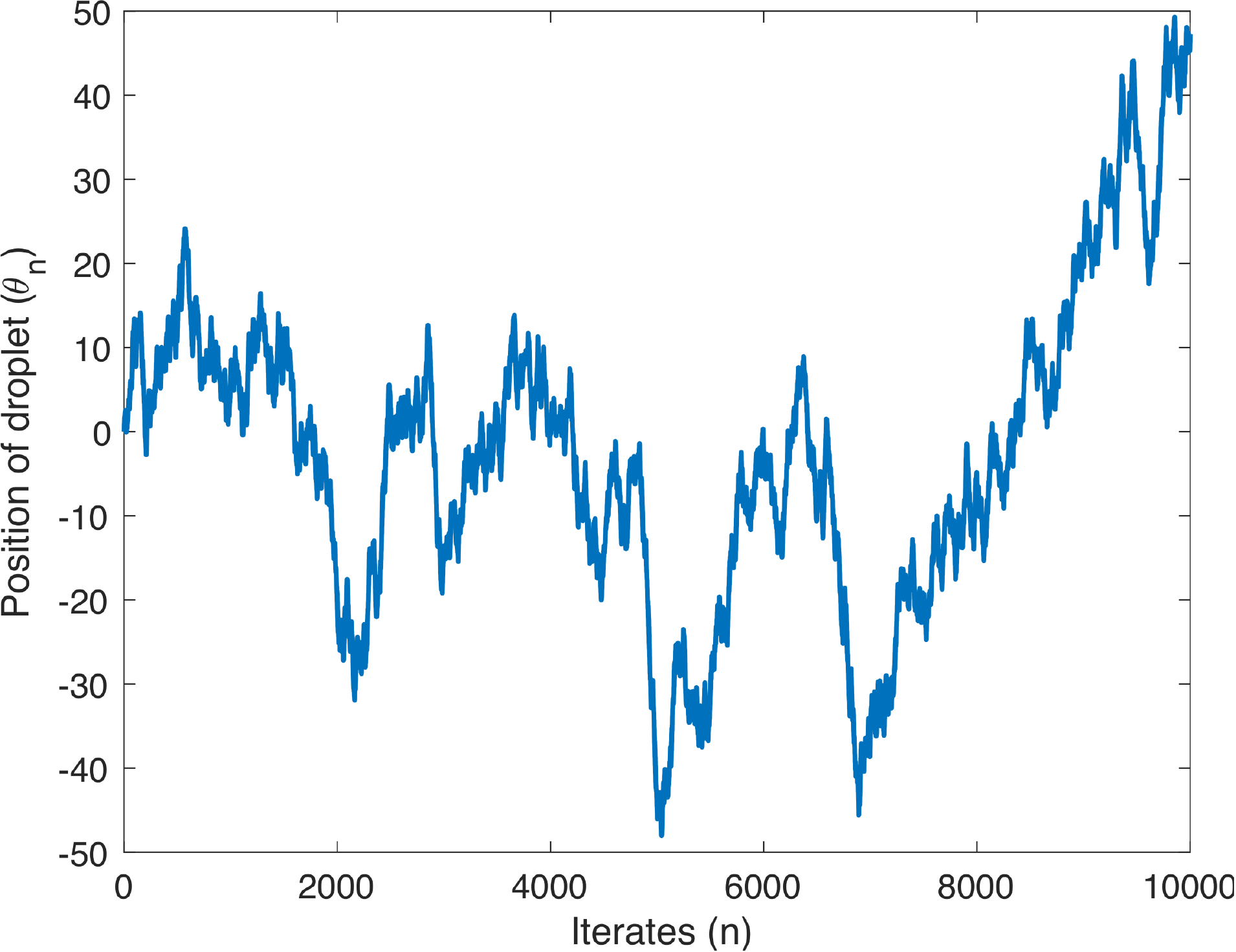}
\caption{Timeseries plot of the drop position in \eqref{Eq: SingleModel} for $C = 1/K$ showing
evidence of chaos.}
\label{Fig: ChaoticTimeSeries}
\end{figure}

\section{Bifurcations and chaos}\label{Sec: Analysis}

Here we analyze \eqref{Eq: SingleMap} the main dynamical model of this investigation.

Let us denote fixed points as $v_*$ and critical bifurcation parameters as $C_*$.
Let us also use the same parameter values for $K$, $\omega$, and $\nu$ as in
Sec. \ref{Sec: Single}.
It is easy to see that $v_* = 0$ is always a fixed point of the system.  For $C \leq C_*$,
zero is the only fixed point (Fig. \hyperref[Fig: 1fp]{\ref{Fig: Pitchfork}a}).  For $C = C_*$
it may seem like Fig. \hyperref[Fig: Linefp]{\ref{Fig: Pitchfork}b} shows a line of fixed
points, however it is in fact just the single fixed point at the origin, which now becomes a saddle.
Finally, for $C > C_*$ two additional fixed points appear (Fig. \hyperref[Fig: 2fp]{\ref{Fig: Pitchfork}c}),
which indicates that there may be a pitchfork bifurcation present.  While it is impossible to solve
for the additional fixed points in closed form, we may still find our bifurcation parameter and
analyze the stability of fixed points by studying the derivative,
\begin{equation}
\label{Eq: Derivative}
f'(v) = C + CK[\omega\cos(\omega v) - 2\nu v\sin(\omega v)]e^{-\nu v^2}.
\end{equation}
Notice that $f'(0) = 1$ when $C = C_* = 1/(1+K\omega)$, which is when the stability of $v_* = 0$
changes from attracting to saddle to repelling.  Further, this value of $C_*$ was confirmed by numerically finding the
bifurcation parameter as the origin changes stability.  This is only possible when the slope of \eqref{Eq: SingleMap}
at $v = 0$ changes from negative to positive, which forces the addition of two fixed points; i.e., a pitchfork bifurcation.

\begin{figure*}[htbp]
\centering

\stackinset{l}{}{t}{}{\textbf{\large (a)}}{\includegraphics[width = 0.32\textwidth]{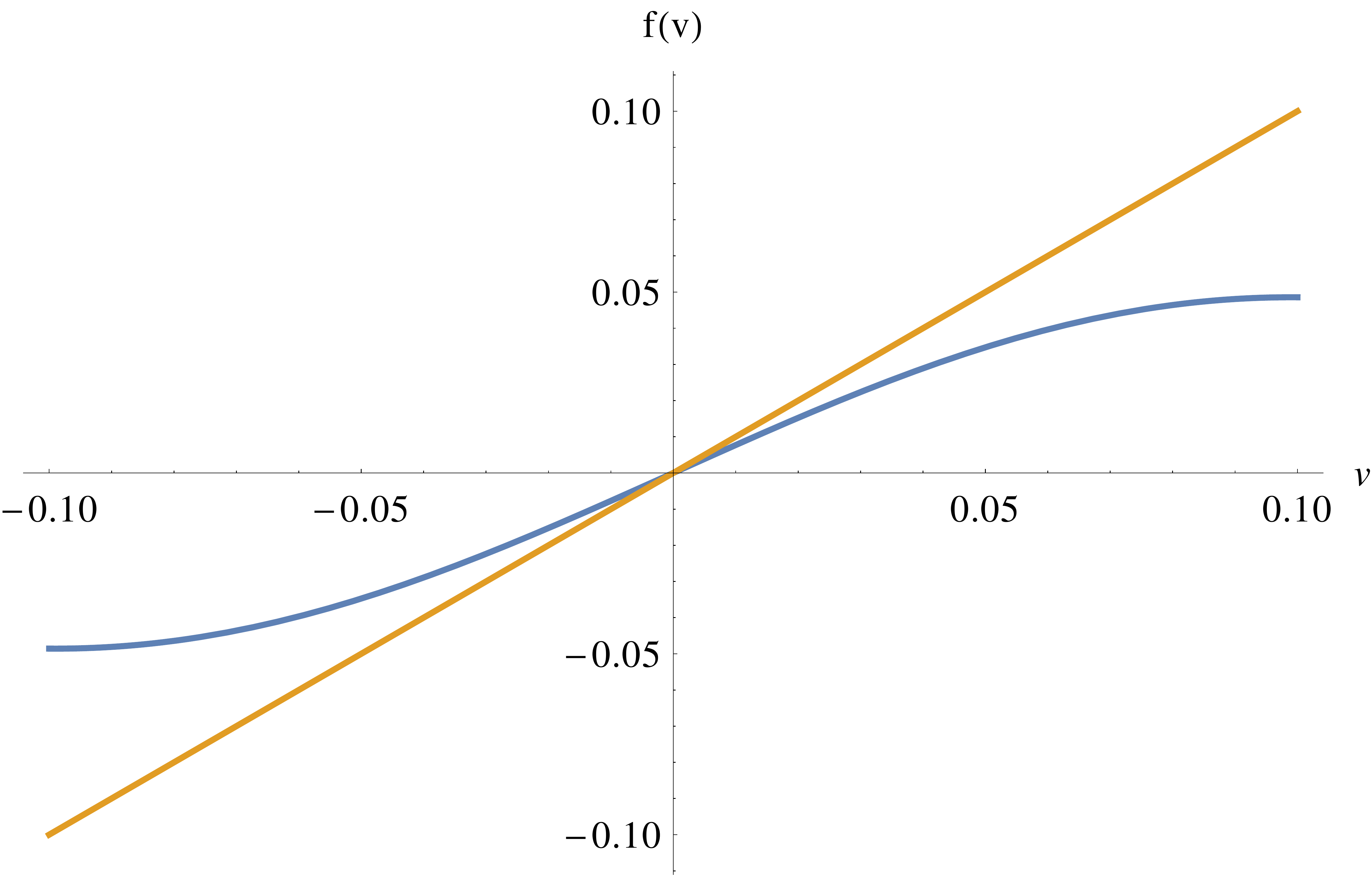}}
\refstepcounter{subfigure}\label{Fig: 1fp}
\stackinset{l}{}{t}{}{\textbf{\large (b)}}{\includegraphics[width = 0.32\textwidth]{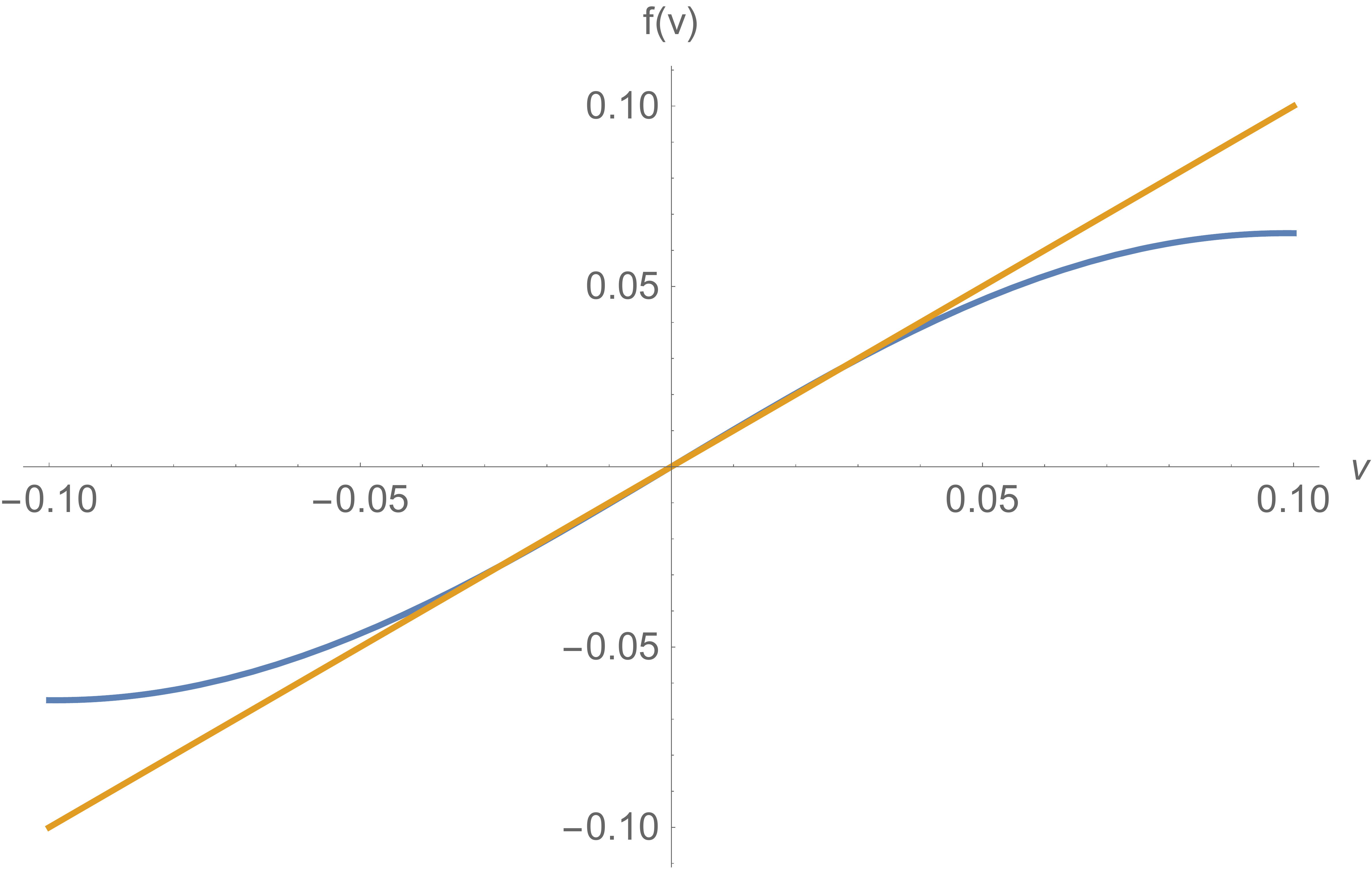}}
\refstepcounter{subfigure}\label{Fig: Linefp}
\stackinset{l}{}{t}{}{\textbf{\large (c)}}{\includegraphics[width = 0.32\textwidth]{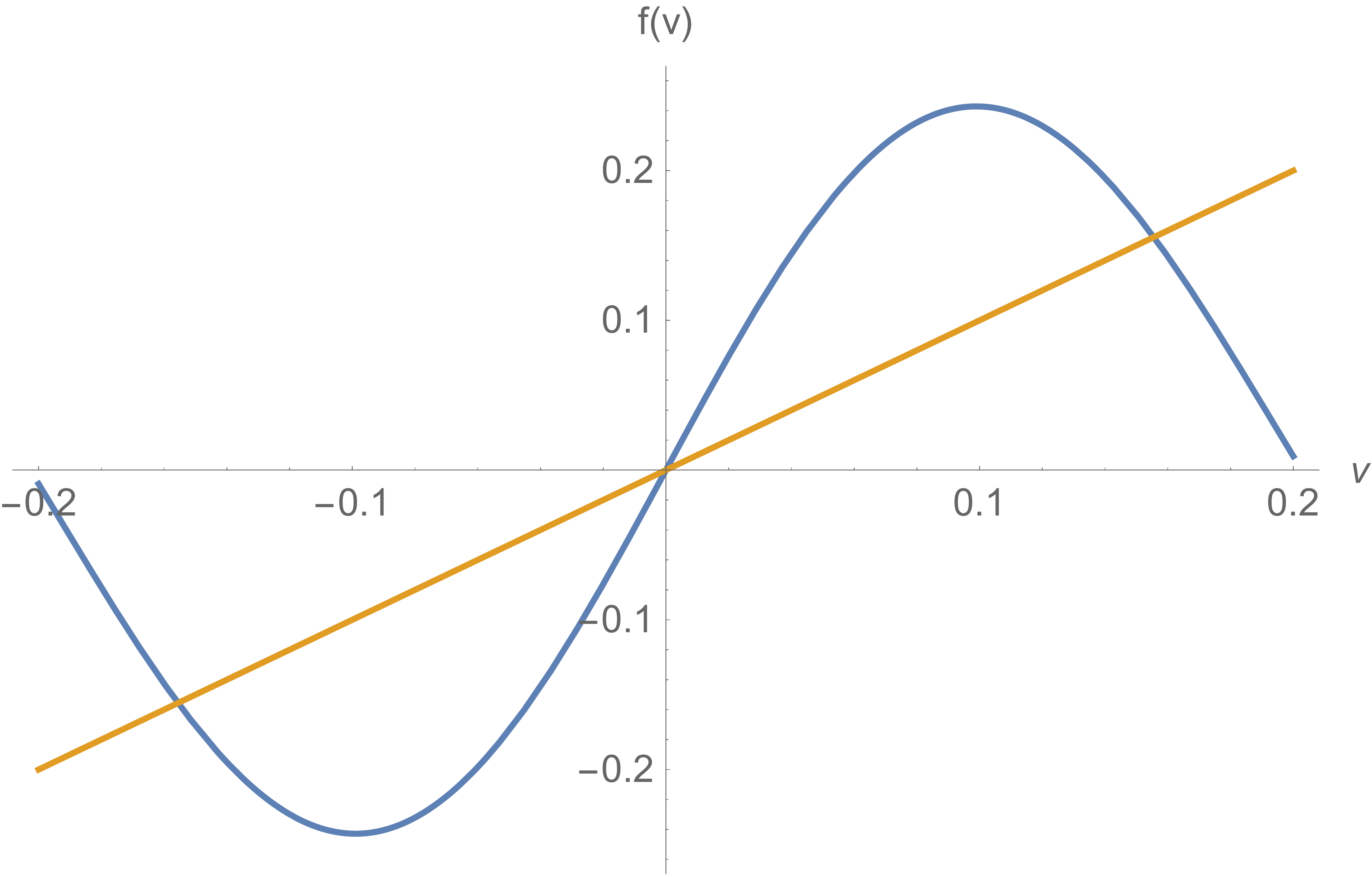}}
\refstepcounter{subfigure}\label{Fig: 2fp}

\caption{As $C$ is increased two additional fixed points appear showing evidence of a pitchfork bifurcation.
Here the blue curve represents $y = f(v)$ and the orange line represents $y = v$.  The intersections indicate
fixed points.  \textbf{(a)}  One attracting fixed point $v_* = 0$ when $C = 1/20K < C_*$.  \textbf{(b)}  One
saddle fixed point $v_* = 0$ when $C = C_* \approx 1/15K$.  \textbf{(c)} 
Two additional fixed points appear when $C = 1/4K > C_*$.}
\label{Fig: Pitchfork}
\end{figure*}

As mentioned it would not be possible to solve for the additional fixed points, but we can prove
their existence, in some parameter regime $C_* < \hat{C} < C < 1$ where
\begin{equation}
\hat{C} = \frac{1}{1 + K\omega \frac{2}{\pi}\exp\left(-\nu \frac{\pi^2}{4\omega^2}\right)}
\end{equation}
by using some basic calculus, which will eventually help us prove the existence of
fixed points of $f^3$; i.e., period-3 orbits.
\begin{lemma}
\label{Lemma: FP}
There exists a fixed point $v_*$ of \eqref{Eq: SingleMap} between $v = \pi/2\omega$ and $v = \pi/\omega$
for $C_* < \hat{C} < C < 1$.  Further, by symmetry, there is also a fixed point $-v_*$.
\end{lemma}

\begin{proof}
Due to symmetry it suffices to prove everything in terms of the positive fixed point.

First we show $f(v) > v$ when $v = \pi/2\omega$.
Notice that $\sin(\omega v) = 1$ when $v = \pi/2\omega$, then
\begin{align*}
f(v) = &C\left[v + Ke^{-\nu v^2}\right] > \frac{v + Ke^{-\nu v^2}}{1 + K\omega \frac{2}{\pi}\exp\left(-\nu \frac{\pi^2}{4\omega^2}\right)}\\
& \geq v\frac{1 + Ke^{-\nu v}/v}{1 + K\omega\frac{2}{\pi}\exp\left(-\nu \frac{\pi^2}{4\omega^2}\right)} = v.
\end{align*}

Next we show $f(v) < v$ when $v = \pi/\omega$.  Notice that $\sin(v) = 0$ when $v = \pi/\omega$, then
$f(v) = Cv < v$ for $C < 1$.  

Since $f \in C([\pi/2\omega,\pi/\omega])$ ($f$ is continuous on the interval
$[\pi/2\omega,\pi/\omega]$), by the intermediate value theorem, $f(v) = v$ for some
$v \in [\pi/2\omega,\pi/\omega]$, thereby completing the proof.
\end{proof}

While the statement of Lemma \ref{Lemma: FP} is weaker than we would like since we require the
use of $\hat{C} > C_*$, it should be noted that for our parameter values
$\hat{C} - C_* \approx 4.8\times 10^{-15} < 2C_*/3$.
It should also be noted that this does not prove uniqueness.
Fortunately, in order to prove the existence of a pitchfork bifurcation we need only analyze the behavior
about $v_* = 0$.

For the next few results we use standard bifurcation analysis as presented in \cite{Kuznetsov}.
Ahead of these results, it is also useful to compute the following derivatives,
\begin{align}
\frac{\partial}{\partial C}f_{C}'(v) = &1 + K[\omega\cos(\omega v)\nonumber \\
&- 2\nu v\sin(\omega v)]e^{-\nu v^2},\label{Eq: Partial_Cv}\\
f''_{C}(v) = &CK\left[(4\nu^2v^2 - 2\nu - \omega^2)\sin(\omega v)\right.\nonumber \\ 
&\left. - 4\nu\omega v\cos(\omega v)\right]e^{-\nu v^2}\label{Eq: SecondDerivative}\\
f'''_{C}(v) = &CK\left[(12\nu^2\omega v^2 - 6\omega v - \omega^2)\cos(\omega v)\right.\nonumber\\
&+(4\nu\omega + 12\nu^2 + 8\nu^2\omega v\nonumber\\
&\left. - 16\nu^4v^3)v\sin(\omega v)\right]e^{-\nu v^2}.\label{Eq: ThirdDerivative}
\end{align}
for which $f'$ represents $\partial f/\partial v$.

\begin{lemma}
\label{Lemma: Pitchfork}
The map \eqref{Eq: SingleMap} is generic about the fixed point $v_* = 0$ and a supercritical pitchfork bifurcation
occurs when $C = C_* = 1/(1+K\omega)$.
\end{lemma}

\begin{proof}
As shown in text $f_{C_*}(0) = 0$ and from \eqref{Eq: Derivative}
$f_{C_*}'(0) = C_*(1+K\omega) = 1$ if $C_* = 1/(1+K\omega)$.
This shows that a 1-parameter bifurcation occurs for \eqref{Eq: SingleMap}.  Now we show that this bifurcation is a
supercritical pitchfork.

We plug $C_* = 1/(1+K\omega)$ and $v_* = 0$ into \eqref{Eq: Partial_Cv} and \eqref{Eq: SecondDerivative}
to get $\partial_Cf'_{C_*}(0) = 0$ and $f''_{C_*}(0) = 0$.  Next we plug our bifurcation parameter and
fixed point into \eqref{Eq: ThirdDerivative},
\begin{equation*}
f'''_{C_*}(0) = C_*K\left[-\omega^2 - 6\omega\nu\right]  < 0,
\end{equation*}
since all of the parameters are positive.  Thereby completing the proof.
\end{proof}

Using similar calculations to that of Lemma \ref{Lemma: Pitchfork}, one may prove a period doubling
bifurcation for the additional fixed point, which was shown to exist in Lemma \ref{Lemma: FP}.
This period doubling may eventually lead to the existence of chaotic orbits by showing the existence
of of a 3-cycle.  We see evidence of period doubling in Fig. \ref{Fig: Doubling}, where 2-cycles are
illustrated.

\begin{figure}[htbp]
\centering
\includegraphics[width = 0.9\textwidth]{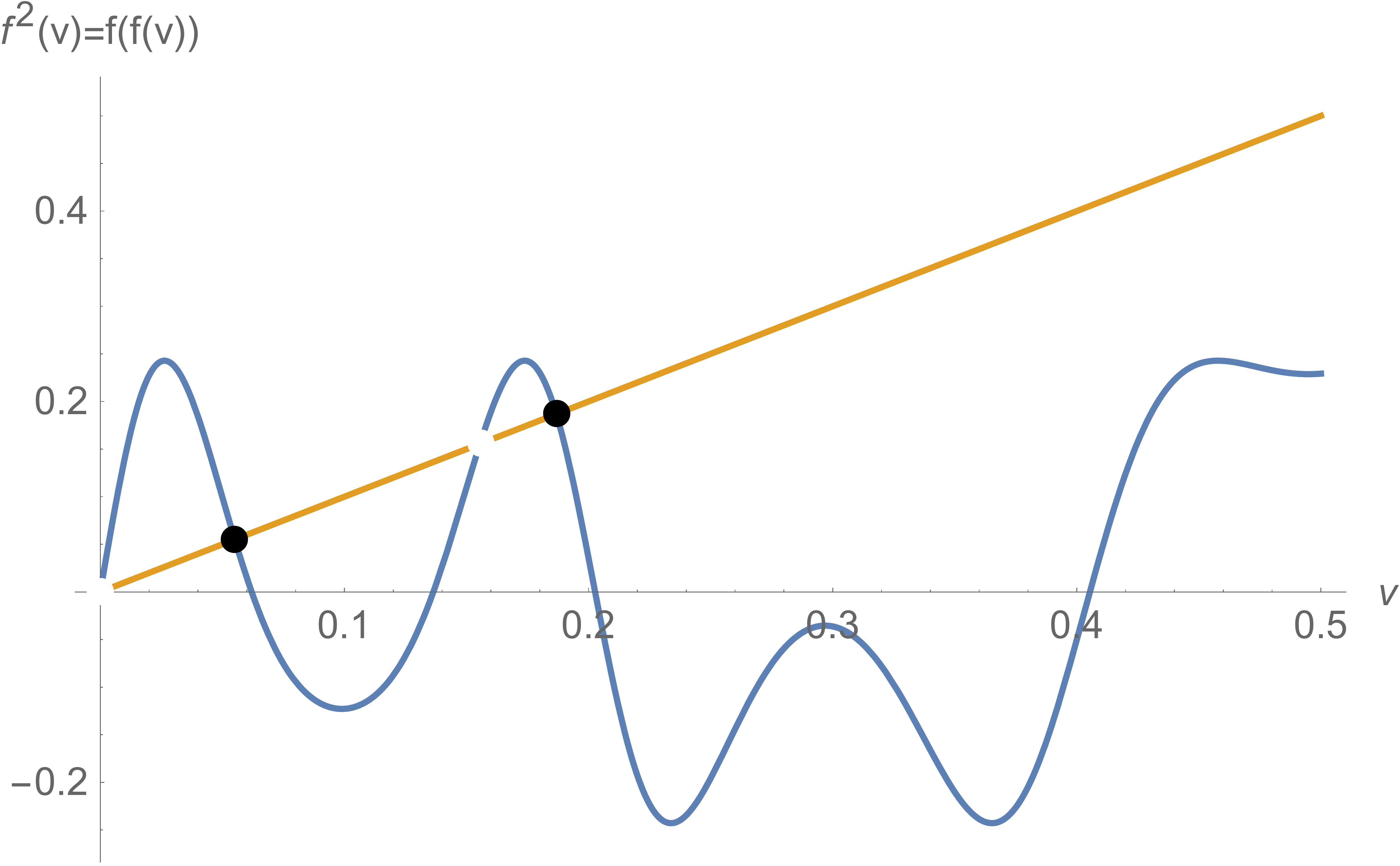}
\caption{In the chaotic regime with $C = 1/4K$, apparent 2-cycles (black markers) are observed, with the
fixed points removed from the plot, providing evidence of a period doubling bifurcation.}
\label{Fig: Doubling}
\end{figure}

For the next theorem we use \eqref{Eq: Partial_Cv}-\eqref{Eq: ThirdDerivative} and refer to
Lemmas \ref{Lemma: FP} and \ref{Lemma: Pitchfork} often.

\begin{thm}
\label{Thm: PeriodDoubling}
There exists a fixed point $v_* \in (\pi/2\omega,\pi/\omega)$ (equivalently for $-v_*$) of the map
\eqref{Eq: SingleMap}, such that $v_*$ undergoes a period doubling bifurcation at the bifurcation
parameter $C = C_{**}$.
\end{thm}

\begin{proof}
It is shown in Lemma \ref{Lemma: FP} that a fixed point $v_*$ exists between $v=\pi/2\omega$
and $v = \pi/\omega$.  Now we need to show that period doubling occurs for some 
$v_* \in (\pi/2\omega,\pi/\omega)$.

By definition the fixed point $v_*$ satisfies $f_{C_{**}}(v_*) = v_*$.  Further, we notice
$f'_{C_{**}}(v_*) = -1$ when
\begin{equation}
C_{**} = -\{1 + K[\omega\cos(\omega v_*) - 2\nu v_*\sin(\omega v_*)]e^{-\nu v_*^2}\}.
\end{equation}
The bifurcation parameter $C_{**}$ can be numerically verified to be bounded as $C_* < C_{**} < 1$.
Since $f'_{C_{**}}(v_*) = -1$, from \eqref{Eq: Partial_Cv}, $\partial_C f'_{C_{**}}(v_*) = -1/C_{**} \neq 0$.

The final condition to satisfy is $\frac{1}{2}f''_{C_{**}}(v_*)^2 + \frac{1}{3}f'''_{C_{**}}(v_*) \neq 0$.
Notice that clearly $\frac{1}{2}f''_{C_{**}}(v_*)^2 > 0$, so then if $f'''_{C_{**}}(v_*) > 0$, the condition
is satisfied.  Since $\pi/2\omega < v_* < \pi/\omega$, $\sin(\omega v) > 0$ and $\cos(\omega v) < 0$.
Further,
\begin{align*}
&12\nu^2\omega v^2 - 6\omega v - \omega^2 < 0, \text{  and}\\
&4\nu\omega + 12\nu^2 + 8\nu^2\omega v - 16\nu^4 v^3 > 0.
\end{align*}
Therefore, $f'''_{C_{**}}(v_*) > 0$, and hence $\frac{1}{2}f''_{C_{**}}(v_*)^2 + \frac{1}{3}f'''_{C_{**}}(v_*) \neq 0$.
Thereby completing the proof.
\end{proof}

We have shown period doubling of a fixed point, which is in fact the beginning of period doubling cascades to
chaos bringing us one step closer to proving the map
\eqref{Eq: SingleMap} is chaotic.  If we could show the existence of a period-3 orbit we would have
a definitive proof of chaos.  This begs the question, ``Does a 3-cycle exist for \eqref{Eq: SingleMap}?''
Fortunately, there is ample evidence of it in Fig. \ref{Fig: 3-cycle}.  Figure \hyperref[Fig: f3]{\ref{Fig: 3-cycle}a}
shows evidence of four 3-cycles (by symmetry) and Fig. \hyperref[Fig: CobWeb]{\ref{Fig: 3-cycle}b}
shows a cobweb plot of a specific 3-cycle.

\begin{figure*}[htbp]
\centering

\stackinset{r}{}{t}{}{\textbf{\large (a)}}{\includegraphics[width = 0.9\textwidth]{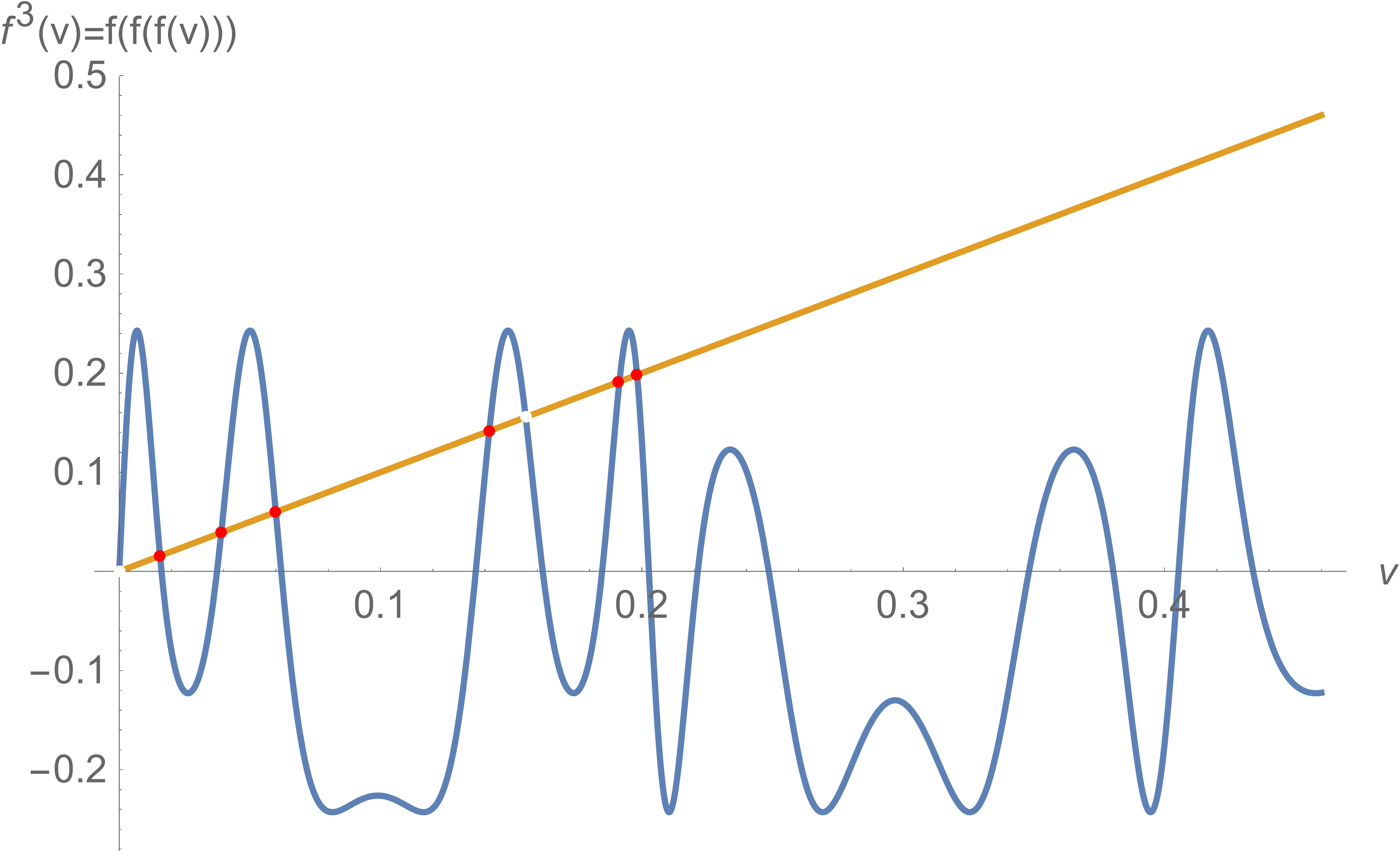}}
\refstepcounter{subfigure}\label{Fig: f3}
\stackinset{r}{}{t}{}{\textbf{\large (b)}}{\includegraphics[width = 0.9\textwidth]{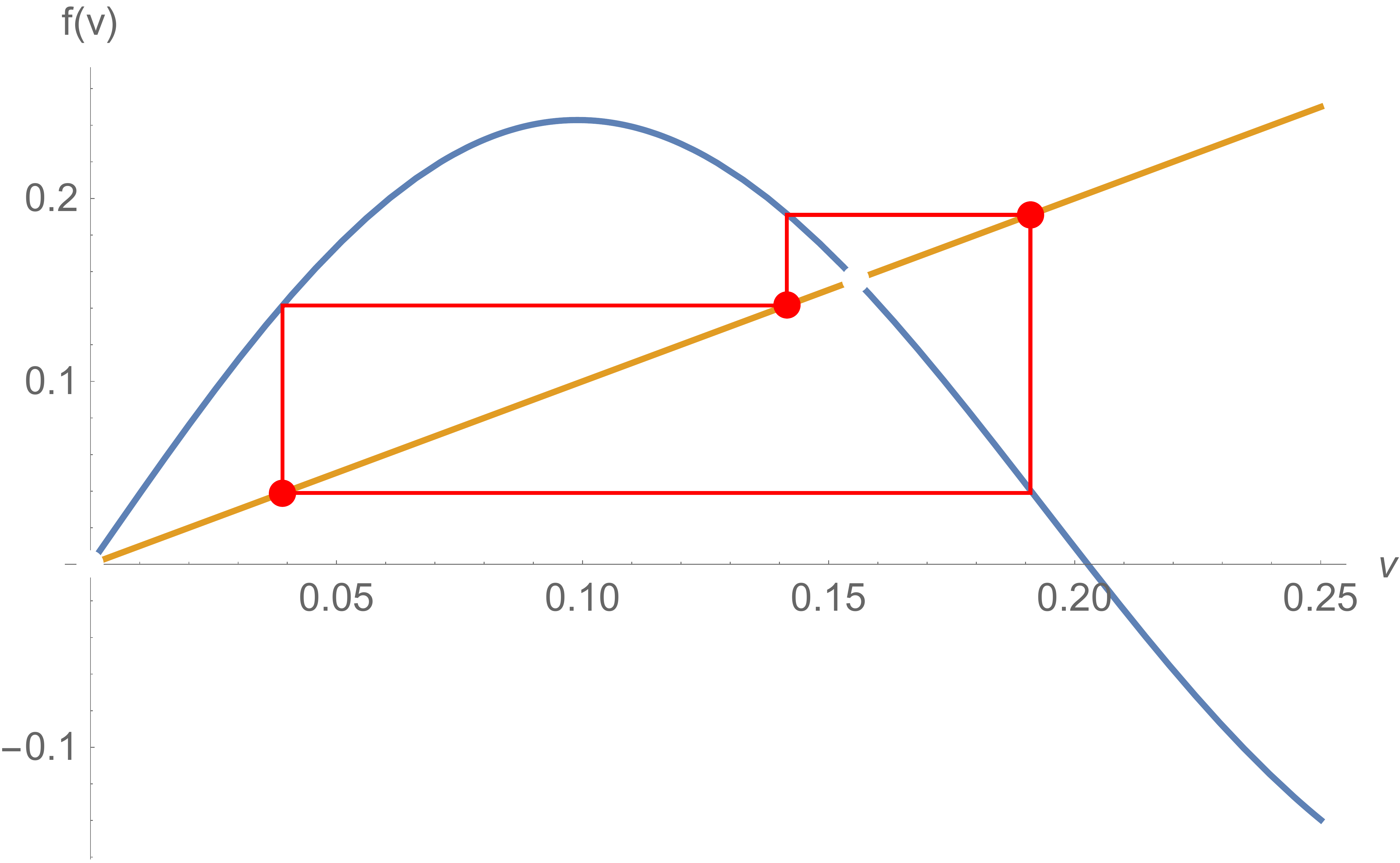}}
\refstepcounter{subfigure}\label{Fig: CobWeb}

\caption{In the chaotic regime with $C = 1/4K$, we observe evident 3-cycles (red markers), which
provides a strong case for the existence of chaotic orbits.  In order to distinguish the 3-cycle points
we remove the fixed points of $f(v)$.  \textbf{(a)}  A set of 3-cycles (in fact, two of them).  By
symmetry this actually shows evidence of four 3-cycles.  \textbf{(b)}  A cobweb plot showing a specific 3-cycle.}
\label{Fig: 3-cycle}
\end{figure*}

To prove the map \eqref{Eq: SingleMap} is chaotic we use the main theorem of Li and Yorke \cite{LiYorke75},
and to satisfy the hypothesis of the theorem we search for 3-cycles.  While a cardinality argument may work
\cite{RJorB18}, it would be quite tricky to keep track of the intersections between $f^3$ and $v$ as $C$
is varied.  While not as satisfactory, it is easier to show the existence of at least one point contained in a
period-3 orbit by using basic calculus techniques to show the existence of a fixed point of $f^3$ different
from the fixed points of $f$.

\begin{thm}
\label{Thm: Chaos}
For every $n \in \N$ there exists periodic points $p_n \in \S$ of the map \eqref{Eq: SingleMap}
$f(v)$ having period $n$ and
$\S$ contains chaotic orbits of $f$.  Furthermore, there exists an uncountable set $S \subset \S$
containing no periodic points.
\end{thm}

\begin{proof}
It was shown by Li and Yorke \cite{LiYorke75} that the existence of a period-3 orbit in a 1-dimensional
map on an interval implies the existence of chaotic orbits on that interval.  Notice that the fact that $\S$
is not an interval causes some problems.  Namely the map $g(\theta) = \theta + \frac{2\pi}{3}$ has
infinitely many 3-cycles (e.g., $\{-2\pi/3,0,2\pi/3\}$), however clearly this map is not chaotic.  So,
let us restrict our domain to $(-\pi/3,\pi/3) \subset \S$ because $|f^3 - v| > 0$ for
$v \in \S\smallsetminus (-\pi/3,\pi/3)$ when $C \leq 1/K$; i.e., there are no 3-cycles outside of this
interval.

In the context of this setup, let us show that $f^3$ has at least one fixed point that is not
a fixed point of $f$.  Due to symmetry and because zero is always a fixed point of $f$,
we may restrict our search to $v \in (0,\pi/3)$.  Furthermore, we choose $C = 1/K$ in order to
simplify computation.

We need to show $f(v) > v$ on some
interval where $f^3(v) = v$.  Notice that if $v < \pi/2\omega < \pi/3$, then $\sin(\omega v) > \omega v/2$
and $e^{-\nu v^2} > 2/\omega$, which give us the inequality
\begin{equation}
\label{Eq: f>v}
f(v)  > C\left[v + K\left(\frac{\omega v}{2}\right)\left(\frac{2}{\omega}\right)\right] = [C+1]v > v.
\end{equation}
So, there are no fixed points of $f$ for $v \in (0,\pi/2\omega)$.

Now we prove that $f^3(v) = v$ for at least one $v$ on $(0,\pi/2\omega)$.  First let us show $f^3(v) > v$
for some $v = v_1$ by
using the inequalities $\sin(\omega v) > \omega v/2$ and $\sin(\omega v) < \omega v$.  We may write an
upper bound for $f^3$ in terms of $f^2$,
\begin{equation}
\label{Eq: Goal}
f^3 > C\left[f^2 + K\left(\frac{\omega f^2}{2}\right)e^{-\nu (f^2)^2}\right]
= \left[C + \frac{\omega}{2}e^{-\nu (f^2)^2}\right]f^2.
\end{equation}
If there exists a $v = v_1$ such that
\begin{equation}
\label{Eq: f^3>f^2}
C + \frac{\omega}{2}e^{-\nu (f^2)^2} > 1 \text{  and  } f^2 > v,
\end{equation}
then we will have $f^3(v_1) > v_1$.  It should be noted that if the first inequality of \eqref{Eq: f^3>f^2}
is satisfied, then $\sin(\omega f^2) > \omega f^2/2$, which is used for \eqref{Eq: Goal}.
Notice that the inequalities of \eqref{Eq: f^3>f^2} imply that
our choice of $v_1$ depends on both an upper bound and lower bound of $f^2$.  By using $\sin(\omega v) < \omega v$
and \eqref{Eq: f>v} we write the upper bound of $f^2$ in terms of $v$,
\begin{subequations}
\label{Eq: f<[C+w]v}
\begin{align}
f(v) &< C\left[v + K(\omega v)e^{-\nu v^2}\right] < [C + \omega]v\\
\Rightarrow f^2(v) &< C\left[f + K(\omega f)e^{-\nu (f)^2}\right] < [C + \omega]^2v.
\end{align}
\end{subequations}
In order to show $f^2 > v$ we have the same criteria as $f^3 > v$, except with substituting
$f$ into \eqref{Eq: f^3>f^2}.  We already have that $f(v) > v$, then our choice of $v = v_1$
would necessitate
\begin{equation}
\label{Eq: f<sqrt}
f < \sqrt{-\frac{1}{\nu}\ln\left(\frac{2}{\omega}\right)} \Rightarrow
C + \frac{\omega}{2}e^{-\nu (f)^2} > 1.
\end{equation}
Similarly \eqref{Eq: f^3>f^2} would also necessitate,
\begin{equation}
\label{Eq: f^2<sqrt}
f^2 < \sqrt{-\frac{1}{\nu}\ln\left(\frac{2}{\omega}\right)} \Rightarrow
C + \frac{\omega}{2}e^{-\nu (f^2)^2} > 1.
\end{equation}
In order to satisfy criteria \eqref{Eq: f<sqrt} and \eqref{Eq: f^2<sqrt} we choose
\begin{equation}
0 < v_1 = \frac{\sqrt{-\ln(2/\omega)/\nu}}{(C+\omega)^2},
\end{equation}
for which \eqref{Eq: f^3>f^2} and \eqref{Eq: f<[C+w]v} imply $f^3(v_1) > v_1$.

Next we show that there exists a $v = v_2$ such that $f^3(v_2) < v_2$ through mainly
brute-force calculations.  Notice that if $f^2 \in (-\pi/2\omega,0)$, then $\sin(\omega f^2) < \omega f^2/2$,
and we write
\begin{equation}
\label{Eq: f^3<v}
f^3 < \left[C+\frac{\omega}{2}e^{-\nu(f^2)^2}\right]f^2 < v,
\end{equation}
since $v$ is positive.  Hence, our choice of $v_2$ must yield $f^2(v_2) \in (-\pi/2\omega,0)$,
so choose
\begin{equation}
v_1 < v_2 = \frac{39\pi}{128\omega} < \frac{\pi}{2\omega},
\end{equation}
then
\begin{align*}
0.805 < f(v_2) = \frac{39\pi C}{128\omega} + \sin\left(\frac{39\pi}{128\omega}\right)\exp
\left(-\nu\frac{39\pi}{128\omega}\right) < 0.81\\
\Rightarrow -\frac{\pi}{2\omega} < f^2(v_2) = Cf(v_2) + \sin(\omega f(v_2))e^{-\nu f(v_2)} < 0.
\end{align*}
This satisfies the criterion for \eqref{Eq: f^3<v}, and therefore shows $f(v_2) < v_2$.

Since $f(v_2) < v_2$ and $f(v_1) > v_1$, by the intermediate value theorem, there exists a
$v_{***} \in (v_1,v_2)$ such that $f^3(v_{***}) = v_{***}$, thereby completing the proof.
\end{proof}

\section{Simple multi-drop velocity model}\label{Sec: Multiple}

In the multi-drop experiments of Filoux \ea\cite{FHV15} it is observed that adding additional
droplets increases the speed of the group, where each of the walkers have the same speed.
Since the velocity always converges to a steady-state, we assume a single drop is at some constant
velocity ($v_{n+1} = v_n$), which is reasonable as mentioned in Sec. \ref{Sec: Single}
and \ref{Sec: Analysis}.  Now,
each additional droplet increases the velocity by changing the wavefield to produce a kick, which
decays over time to zero.  Let us assume it has a kick strength $\kappa$ similar to that of $C\cdot K$ in
\eqref{Eq: SingleModel}, and a similar type of spatial decay with parameter $\mu$, where the farther the
source of the kick is the less effective it is.

\begin{figure}[htbp]
\centering
\includegraphics[width = 0.9\textwidth]{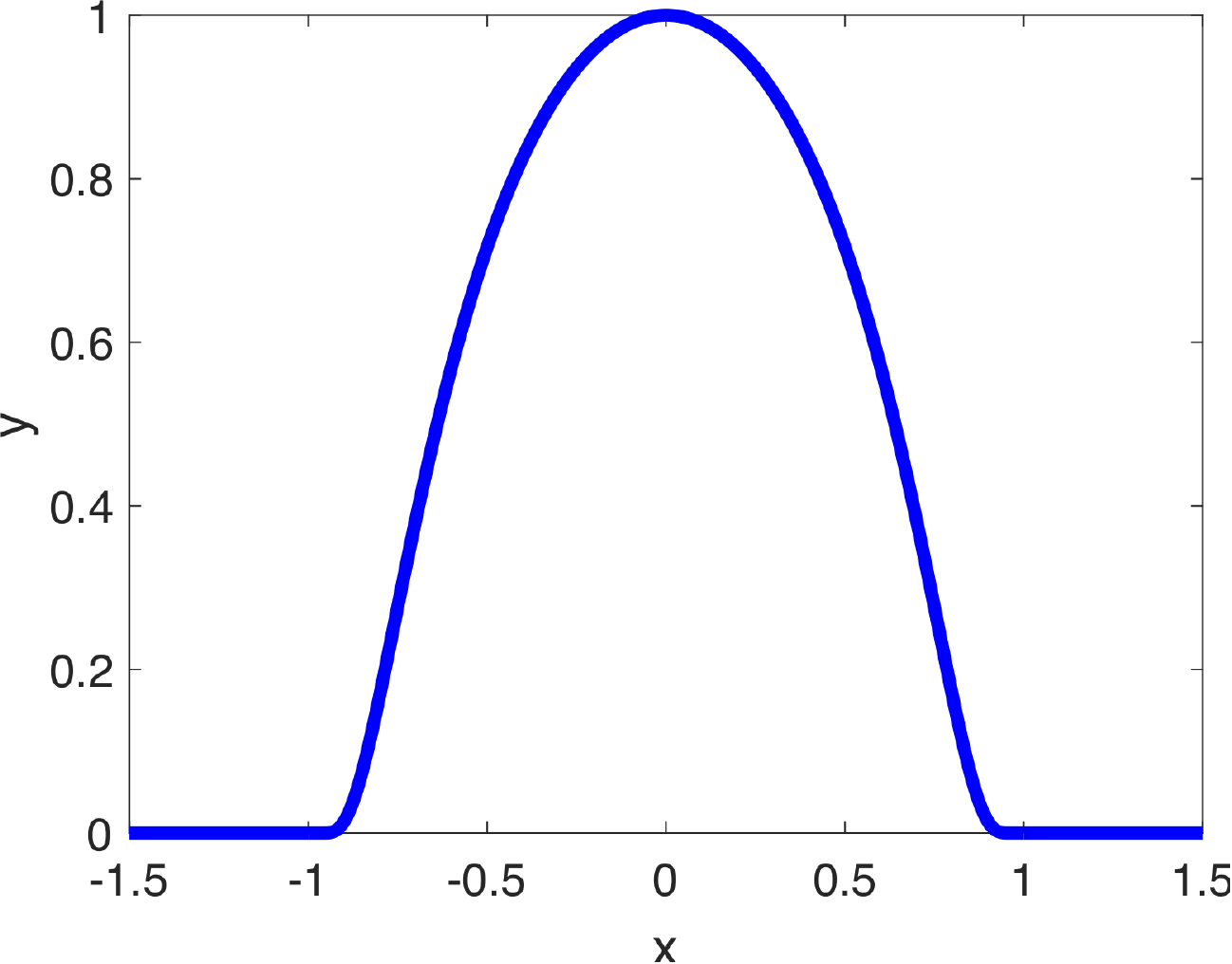}
\caption{An example of a generic bump function.}
\label{Fig: Bump}
\end{figure}

For a single droplet the source of the kick was itself, however for this model we assume the kick from
itself is absorbed into the constant velocity and we consider the kicks from the additional droplets only,
with each droplet providing a different kick.  Further, there is a vanishing temporal decay for which the
earlier impacts produce a smaller decay than latter ones.  Then we write
\begin{equation}
\begin{split}
v_{n+1} &= v_n + \kappa\eta(\gamma n)\sum_{m=1}^M e^{-\mu\big|\theta_n - \theta_n^m\big|}\\
\eta(\xi) &:= \begin{cases}
\exp\left(1 - \frac{1}{1 - \xi_n^2}\right) & \xi_n \in (-1,1),\\
0 & \xi_n \notin (-1,1);
\end{cases}
\end{split},
\label{Eq: Multiple}
\end{equation}
where $v_n$ and $\theta_n$ are the velocity and position of the original droplet after impact $n$,
$\theta_n^m$ is the position of the $m^\th$ additional droplet, and
$M$ is the total number of additional droplets.  The one aspect quite different from \eqref{Eq: SingleModel}
is the use of $\eta$ to represent the temporal damping.  While different functions may be used,
a bump function (Fig. \ref{Fig: Bump}) contains all of the observed properties.  For the argument of
$\eta$, $n$ is the number of impacts since the addition of the last droplet, and $\gamma\in [0,1]$ is the
strength of decay.  If $\gamma = 0$, the kick does not decay exponentially, although it does decay linearly
due to the use of $\kappa$.  Further, if $\gamma = 1$ there is no kick; 
otherwise, the kick decays as $n = \lceil 1/\gamma \rceil$.

We also use $|\theta_n - \theta_n^m|$
instead of $(\theta_n - \theta_n^m)^2$ in our spatial damping term.  In Sec. \ref{Sec: Single} a square was used
in order to make \eqref{Eq: SingleMap} smooth for the sake of analysis, and it did not affect the qualitative behavior
of the droplet.  For multiple drops, however, a square term would be unnecessary because detailed dynamical systems
analysis would be too involved for this study.  In addition, the inclusion of
$(\theta_n - \theta_n^m)^2$ in the argument of the exponential
creates such a difference in velocity contribution between the $m^\th$ and $(m+1)^\th$ additional droplet that
the simulations do not match the experimental data quantitatively, which is the goal of this model.  Since we use
an absolute value instead of a square for the spatial term, we will have a spatial damping parameter $\mu$
different from $\nu$, namely $\mu = \omega/2.1\cdot 2\pi$.

\subsection{Simulations and comparisons}

In this section, we compare the velocity from \eqref{Eq: Multiple} with
the experiments of Filoux \ea \cite{FHV15, FHV15Arxiv}.
For each of the simulations we use $\omega = 31/2$, which is within the experimental range of
Filoux \ea \cite{FHV15}.  Furthermore, the axes tick marks of the figures are matched up exactly
in order to facilitate proper comparison.

In Fig. \ref{Fig: MyWalkers}, the velocity for each successive droplet from \eqref{Eq: Multiple} is
plotted on top of the experimental data.  We choose a kick parameter of $\kappa = 1/7$, which
is below the observed threshold for chaos in \eqref{Eq: SingleModel}.  In addition, a temporal damping
factor of $\gamma = 0.44$ is used, which implies that the system reaches a steady state within $n=3$
bounces.  As we observe, the model \eqref{Eq: Multiple} exhibits
extremely close agreement with the experimental data.

\begin{figure}[htbp]
\centering
\includegraphics[width = 0.9\textwidth]{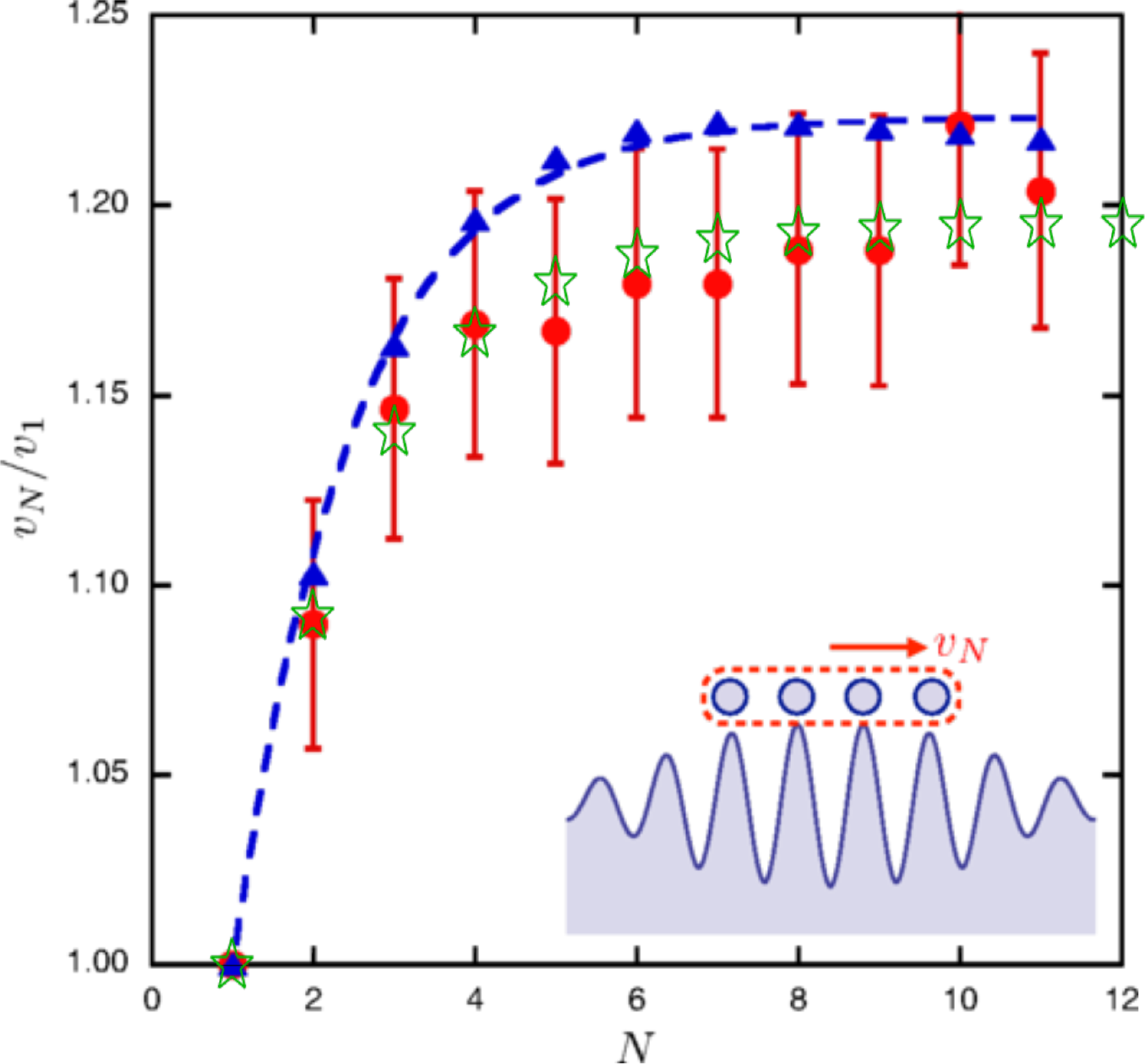}
\caption{Speed of each additional walker.  Green stars represent \eqref{Eq: Multiple}, red markers
are experimental data points with error bars\cite{FHV15}, and blue triangles are from the algebraic model of
Filoux \ea \cite{FHV15} with the dashed curve fitting the points from their model.  For the green stars
$\gamma = 0.44$ and $\kappa = 1/7$ from \eqref{Eq: Multiple}.  Further, the distance between droplets,
$s = \theta - \theta^1 = (8/6)\cdot 2\cdot \pi/\omega \approx 8 mm$, similar to the spacing used in
\cite{FHV15}.  The markers from the results of this article were embedded directly onto the original
figure by precisely matching the axes, and without further modification or manipulations.  {\color{red}
Figure adapted and modified with permission from Filoux \ea \cite{FHV15}}.}
\label{Fig: MyWalkers}
\end{figure}

Filoux \ea also presented plots of the velocity depending on the separation distance of a pair of droplets
\cite{FHV15,FHV15Arxiv}.  Figure \hyperref[Fig: MyWalkers1]{\ref{Fig: MyWalkers12}a} shows the ratio
of speeds of a pair of droplets without offsets or rescaling of the axes\cite{FHV15Arxiv}.  Figure 
\hyperref[Fig: MyWalkers2]{\ref{Fig: MyWalkers12}b} shows the ratio of speeds of a pair of droplets with
an offset of unity and a log-scale for the ordinate\cite{FHV15}.    In order to show the robustness of the
model \eqref{Eq: Multiple}, we use the same parameters as Fig. \ref{Fig: MyWalkers} to plot the ratio of
velocities from our model on top of this
set of experimental data.  Once again the model exhibits extremely close agreement with experiments.

\begin{figure}[htbp]
\centering

\stackinset{r}{0.5mm}{t}{1mm}{\textbf{\large (a)}}{\includegraphics[width = 0.44\textwidth]{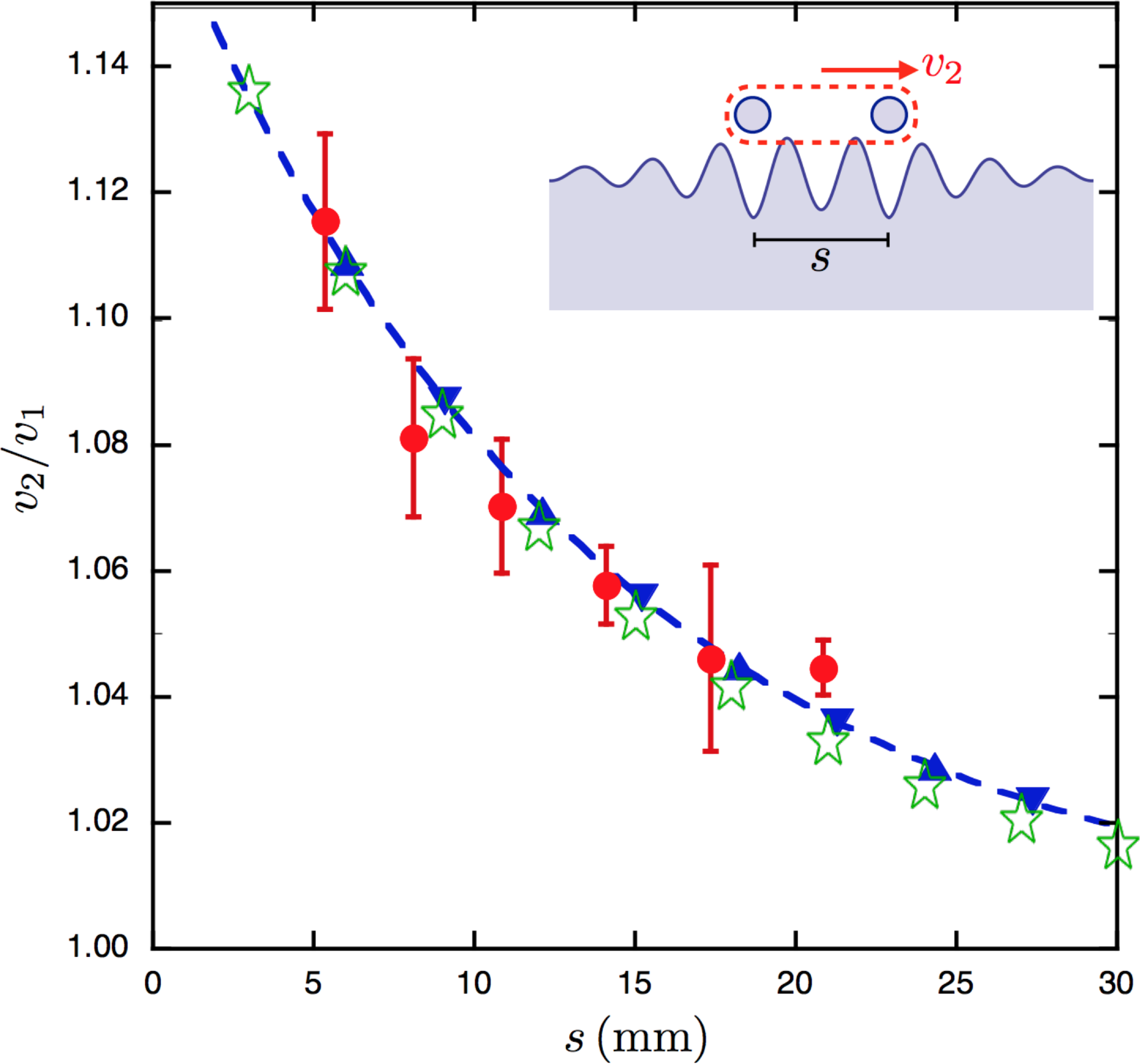}}
\refstepcounter{subfigure}\label{Fig: MyWalkers1}
\stackinset{r}{0.5mm}{t}{1mm}{\textbf{\large (b)}}{\includegraphics[width = 0.44\textwidth]{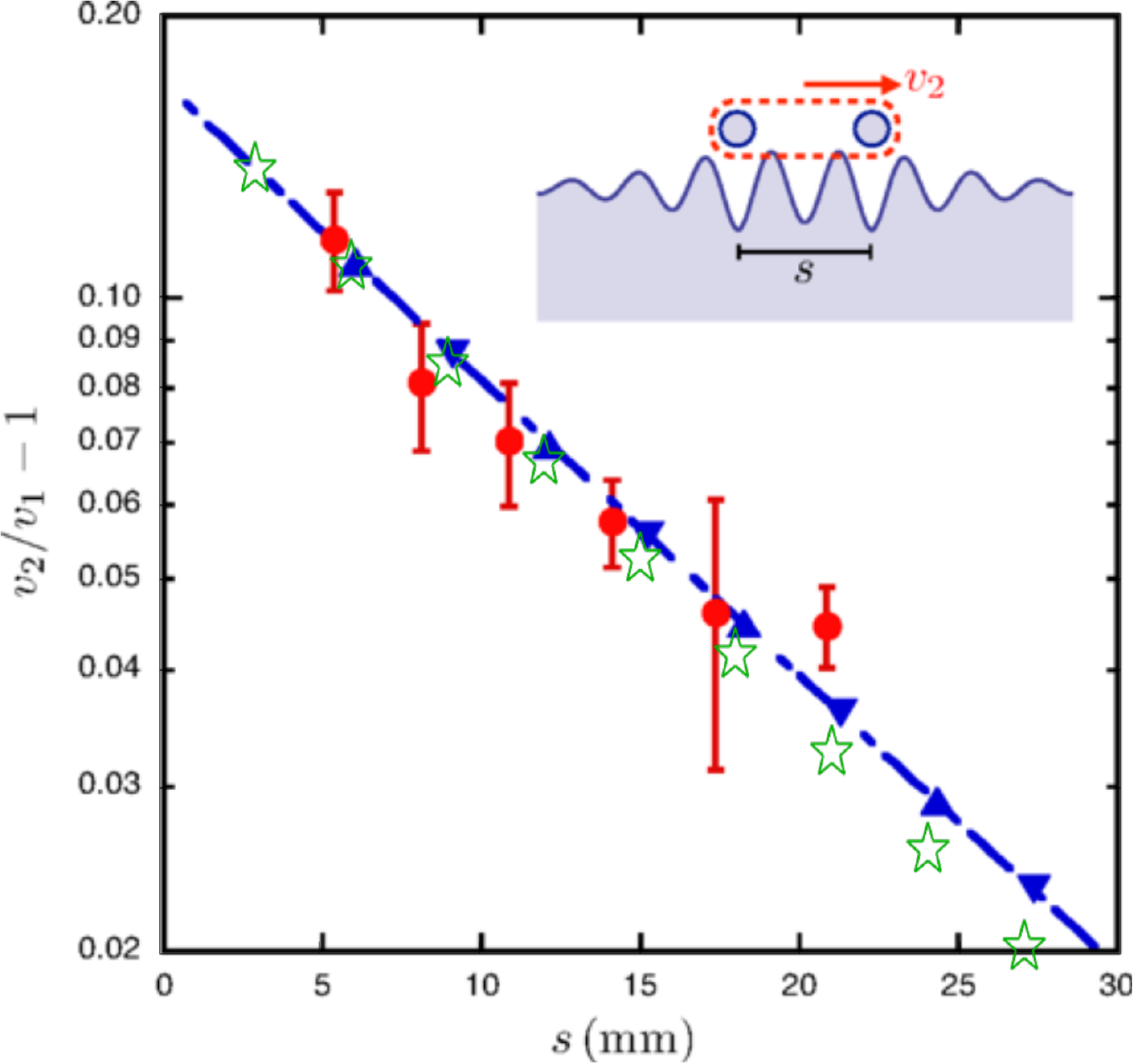}}
\refstepcounter{subfigure}\label{Fig: MyWalkers2}

\caption{Ratio of speeds of a pair of droplets as the separation is increased.  Green stars represent
\eqref{Eq: Multiple}, red markers are experimental data points with error bars\cite{FHV15}, and blue 
triangles are from the algebraic model of Filoux \ea \cite{FHV15} with the dashed curve fitting the
points from their model.  For the green stars $\gamma = 0.44$ and $\kappa = 1/7$ from \eqref{Eq: Multiple}.
\textbf{(a)}  Ratio of speeds without offsets or rescaling
of axes as presented in the original figure of Filoux \ea \cite{FHV15Arxiv}.
\textbf{(b)}  Ratio of speeds offset by unity with a log-scale for the ordinate as presented in the
original figure of Filoux \ea \cite{FHV15}.
The markers from the results of this article were embedded directly onto the original
figures by precisely matching the axes, and without further modification or manipulations.  {\color{red}
Figures adapted and modified with permission from Filoux \ea \cite{FHV15, FHV15Arxiv}}.}
\label{Fig: MyWalkers12}
\end{figure}

\pagebreak

\section{Potential unified model}\label{Sec: Combined}

To predict unforeseen behavior we derive a unified model and present simulations predicting behavior of
multiple droplets in the chaotic regime.  In Sec. \ref{Sec: Multiple} it was assumed that a single droplet
would move with a constant speed, however as shown in Sec. \ref{Sec: Single} this is not the case when
the damping is low enough.  So, we may simply include the right hand side of the velocity equation from
\eqref{Eq: SingleModel} in place of the constant velocity in \eqref{Eq: Multiple}.  However, in order to use
the same parameters for the two parts, we write the argument of the exponent for the spatial damping
with $\left(\theta_n^i - \theta_n^m\right)^2$ instead of $\big|\theta_n^i - \theta_n^m\big|$.
Further, in general we expect droplets to repel each other as the distance between two successive ones
decrease, so we include a signum function as the coefficient of the exponential term,
\begin{equation}
v_{n+1}^i = 
\begin{cases}
C\left[v_n^i + K\sin\left(\omega v_n^i\right)e^{-\nu(v_n^i)^2}\right. & \\
\left. + K\eta(\gamma n)\sum_{m\neq i}^M \sgn\left(\theta_n^i - \theta_n^m\right)
e^{-\nu\left(\theta_n^i - \theta_n^m\right)^2}\right] & M \neq 0,\\
C\left[v_n^i + K\sin\left(\omega v_n^i\right)e^{-\nu(v_n^i)^2}\right] & M = 0.
\end{cases}
\label{Eq: Combined}
\end{equation}
In \eqref{Eq: Combined} $v_n^i$ and $\theta_n^i$ are the velocity and position of the $i^\th$ droplet after
the $n^\th$ impact.  The parameters in \eqref{Eq: Combined} are the same as \eqref{Eq: SingleModel}, but
different from \eqref{Eq: Multiple}.

While \eqref{Eq: Combined} is a speculative model and requires more detailed scrutiny and modifications,
it exhibits three types of behavior (Fig. \ref{Fig: MultiChaotic}) that may potentially be observed in future experiments.
We use five droplets in Fig. \ref{Fig: MultiChaotic} to simulate \eqref{Eq: Combined}.
In Fig. \hyperref[Fig: MultiChaotic6k]{\ref{Fig: MultiChaotic}a} we observe regular motion of the droplets
around the annulus where the distance between successive droplets remain nearly constant.  Then
Fig. \hyperref[Fig: MultiChaotic499k]{\ref{Fig: MultiChaotic}b} shows destabilization of the droplet
configuration leading to successive droplets approaching each other, but without crossing paths.
Figure \hyperref[Fig: MultiChaotic3k]{\ref{Fig: MultiChaotic}c} shows complete destabilization with
droplets crossing paths, or perhaps going around one another.

\begin{figure*}[htbp]
\centering

\stackinset{l}{4.8mm}{t}{1mm}{\textbf{\large (a)}}{\includegraphics[width = 0.32\textwidth]{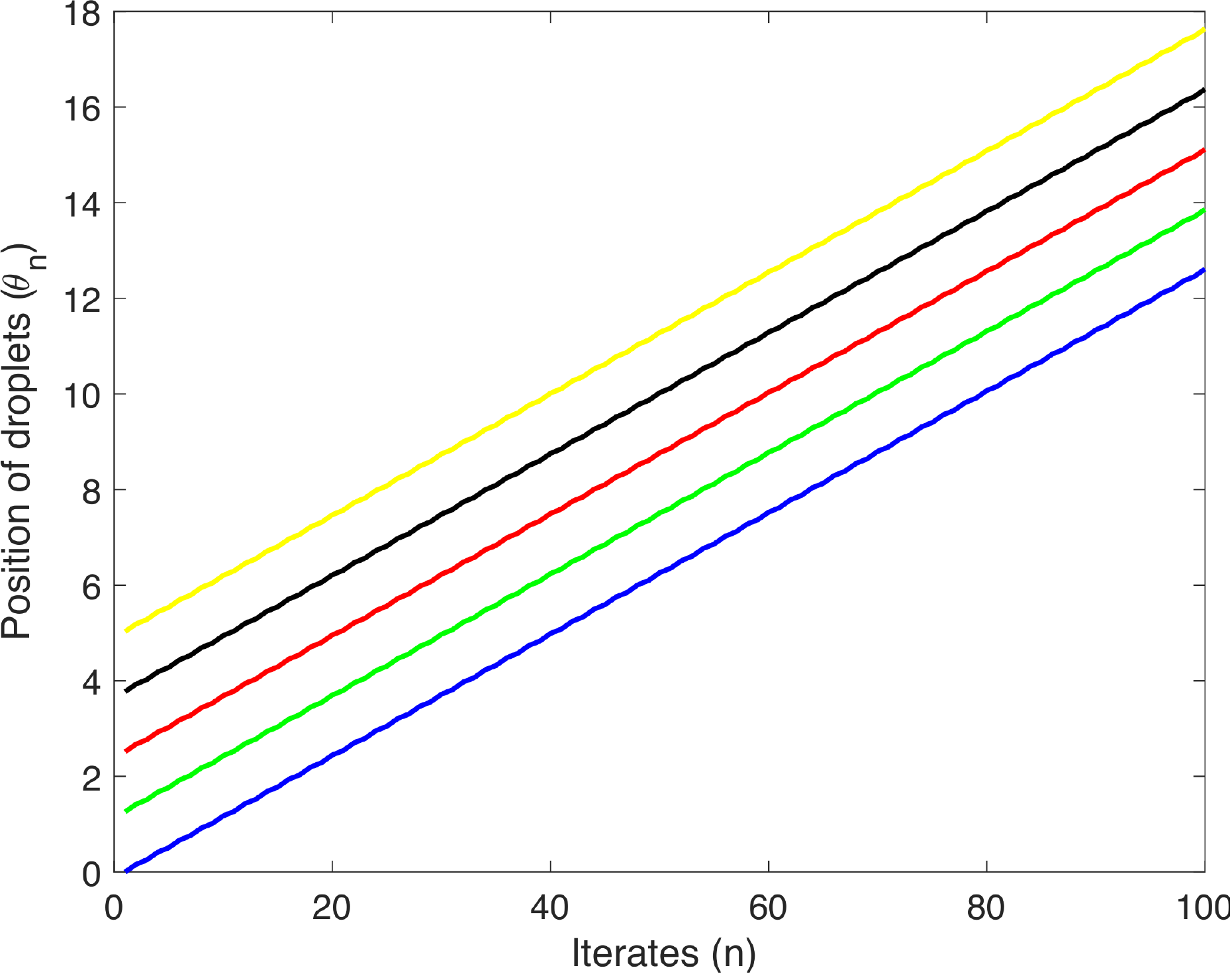}}
\refstepcounter{subfigure}\label{Fig: MultiChaotic6k}
\stackinset{l}{4.8mm}{t}{1mm}{\textbf{\large (b)}}{\includegraphics[width = 0.32\textwidth]{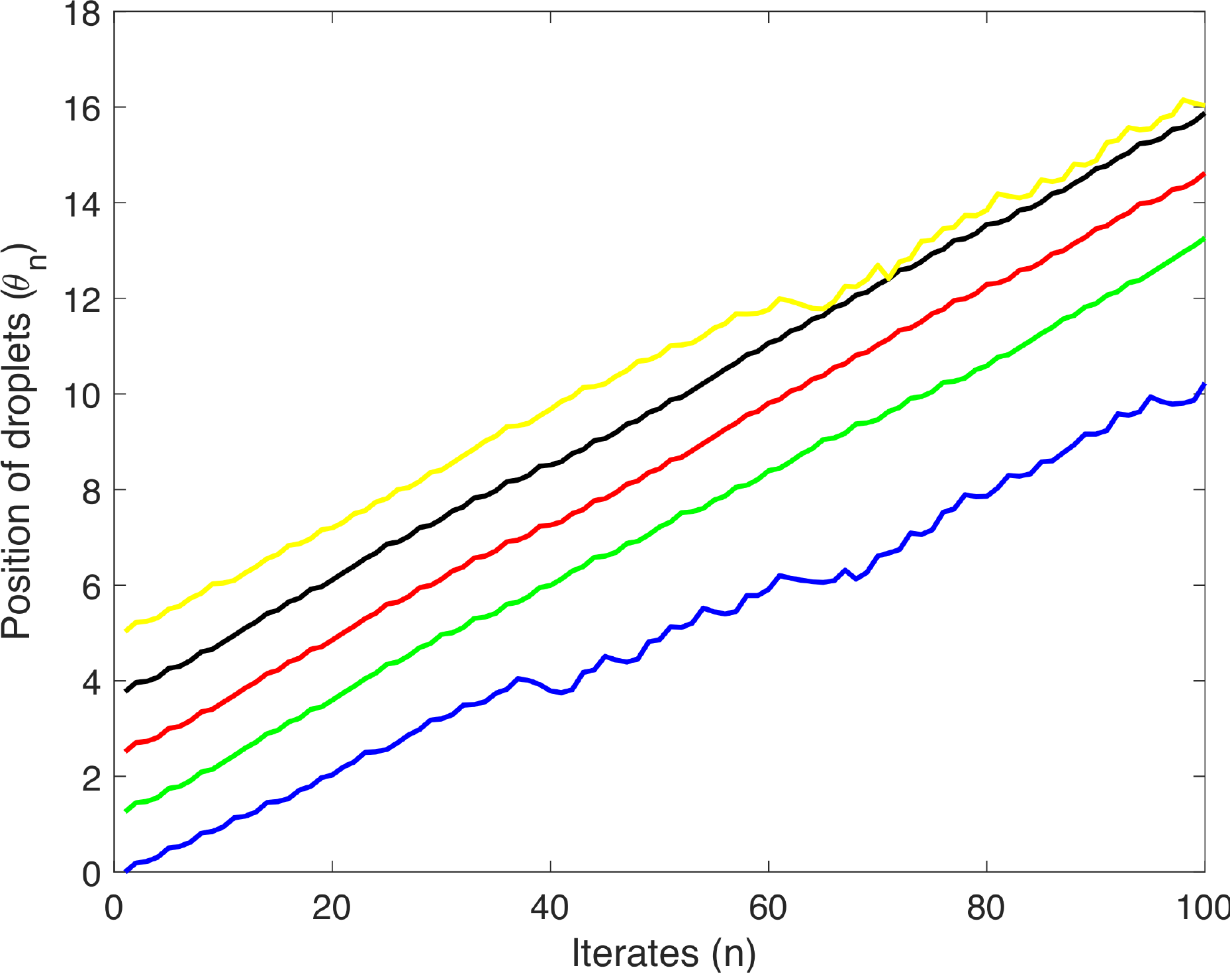}}
\refstepcounter{subfigure}\label{Fig: MultiChaotic499k}
\stackinset{l}{4.8mm}{t}{1mm}{\textbf{\large (c)}}{\includegraphics[width = 0.32\textwidth]{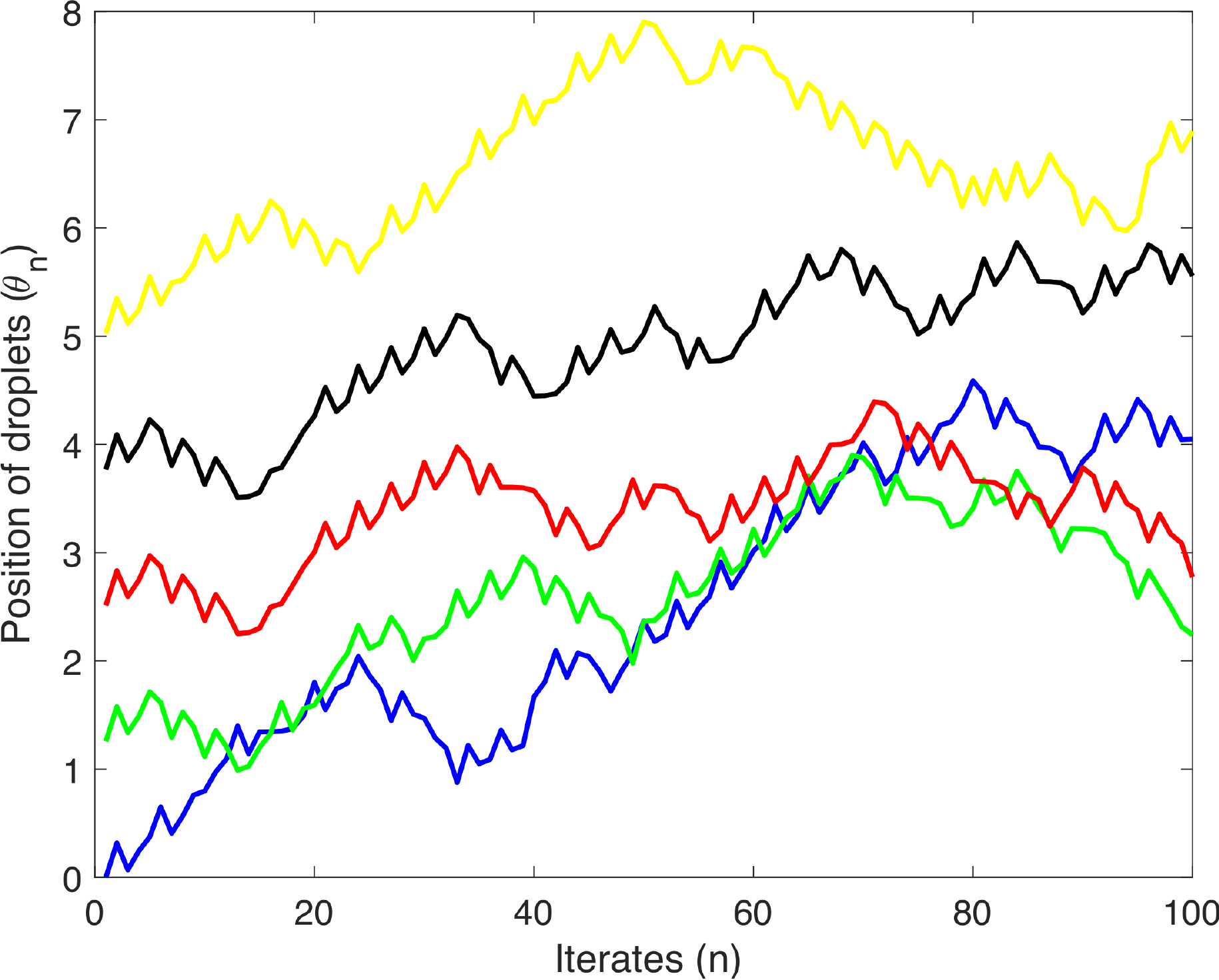}}
\refstepcounter{subfigure}\label{Fig: MultiChaotic3k}

\caption{Three different types of observations from \eqref{Eq: Combined}.  The colors in each figure signifies a
different droplet.  \textbf{(a)}  The droplets walk in an orderly fashion with the distance between successive
droplets remaining nearly constant.  \textbf{(b)}  The droplet configuration destabilizes and the distance varies,
but the droplets do not go past one another.  \textbf{(c)}  The droplets start going past one another.}
\label{Fig: MultiChaotic}
\end{figure*}

\section{Conclusion}
\label{Sec: Conclusion}

While there have been many detailed hydrodynamic models of walking droplet phenomena, simple
discrete dynamical models are a recent occurrence.  The first model amenable to rigorous dynamical
systems analysis without additional simplifications, developed by Gilet, described a droplet walking
on a straight line\cite{Gilet14}.  Although it exhibits interesting dynamical behavior, there are limitations
in comparing it with experiments due to complications at boundaries of a finite rectangular domain.
One way to remove the boundary complication is to study the walker on an annulus.  Fortunately,
there have been experiments conducted on walking in an annular cavity\cite{FHV15, FHV15Arxiv,
PucciHarrisPrivate, MilesPrivate}.

In this investigation we developed discrete dynamical models of a single walker on an annulus in
Sec. \ref{Sec: Single} and multiple non-chaotic walkers in Sec. \ref{Sec: Multiple}.  The single
droplet model is simulated numerically and analyzed rigorously via dynamical systems theory
in Sec. \ref{Sec: Analysis}.  Through the numerical simulations qualitative agreement with the
experiments of Pucci and Harris \cite{PucciHarrisPrivate} is shown.  Further, the simulations
indicate the existence of pitchfork bifurcations, period doubling bifurcations, and chaos.  The
rigorous dynamical systems analysis of Sec. \ref{Sec: Analysis} proved the existence of additional
fixed points (Lemma \ref{Lemma: FP}), the appearance of which is proved to be a pitchfork bifurcation
(Lemma \ref{Lemma: Pitchfork}).  Then the existence of a period doubling bifurcation for those fixed points
is proved in Theorem \ref{Thm: PeriodDoubling}.  Chaos is proved by showing the existence of
period-3 orbits (Theorem \ref{Thm: Chaos}) by using the result of Li and Yorke \cite{LiYorke75}.
The multiple non-chaotic walker model from Sec. \ref{Sec: Multiple} is simulated showing
excellent agreement with several experimental data sets of Filoux \ea \cite{FHV15, FHV15Arxiv}.
Finally, a speculative unified model is derived in Sec. \ref{Sec: Combined} with simulations showing
plausible scenarios for future experiments.

While much of the behavior observed in experiments is captured in the models of Secs. \ref{Sec: Single}
and \ref{Sec: Multiple}, there are still subtleties that need to be investigated.  One such issue is
the adverse effect of using $\left(\theta_n - \theta_n^m\right)^2$ instead of $\big|\theta_n - \theta_n^m\big|$
has on the simulations.  This small difference, even while being consistent with the parameters for the
change, causes the simulations to stray far from quantitative experimental observations. In addition,
the analysis done in this article can be extended further in a future work focusing on dynamical systems
and bifurcations of the model(s).  Moreover, the speculative unified model of Sec. \ref{Sec: Combined}
needs to be scrutinized in detail and modified as more experimental data becomes available.
We endeavor to tackle these issues in forthcoming works.

\section*{Acknowledgment}
First of all, the author would like to express his gratitude to the editors of the Hydrodynamic Quantum Analogs
special issue, and specifically J. Bush for the invitation to contribute to a topic of great scientific and mathematical
importance.  Thanks are also due D. Blackmore for enlightening discussions and V. Ratnaswamy for helping recover
an accidentally deleted figure file.  Finally, the author wishes to thank the Department of Mathematics and Statistics
at TTU and the Department of Mathematical Sciences at NJIT for their support.

\bibliographystyle{unsrt}
\bibliography{Bouncing_droplets}

\end{document}